%% file: paper.tex
\documentclass[12pt,draftclsnofoot,peerreviewca,onecolumn,a4]{IEEEtran}

\usepackage{amsmath, mathrsfs}
\usepackage{mathtools}
\usepackage{bm}
\usepackage{bbm}
\usepackage{stmaryrd}
\usepackage[dvipsnames]{xcolor}
\makeatletter
\def\maketag@@@#1{\hbox{\m@th\normalfont\normalsize#1}}
\makeatother

\newcommand{\subparagraph}{}
\usepackage[compact]{titlesec}
\titlespacing*{\section}{0pt}{1\baselineskip}{0.9\baselineskip}
\allowdisplaybreaks
\usepackage{tikz}
\usepackage{pgfplots}
\pgfplotsset{compat=newest}
\usetikzlibrary{patterns}
\usetikzlibrary{positioning}
\usetikzlibrary{datavisualization}
\usetikzlibrary{datavisualization.formats.functions}
\usetikzlibrary{backgrounds}
\usetikzlibrary{shapes,snakes}
\usepackage{textcomp}
\usetikzlibrary{positioning}

\usepgfplotslibrary{external} 
\usepackage{cite}
\usepackage{amsfonts}
\usepackage{amssymb}
\usepackage{color}
\usepackage{multirow}
\usepackage{graphicx}
 \usepackage{relsize}
\usepackage{multirow}
\usepackage{multicol}
\usepackage{caption}
\usepackage{subcaption}
\usepackage[numbers,sort&compress]{natbib}
\usepackage{balance}
\usepackage{arydshln}

\usepackage{pgfplots} 
\usepackage{pgfgantt}
\usepackage{pdflscape}
\pgfplotsset{compat=newest} 
\pgfplotsset{plot coordinates/math parser=false}
\newlength\fwidth
\input{symbols.tex}

\input{macros.tex}

\usepackage[nodisplayskipstretch]{setspace}

\usepackage{algorithm}
\usepackage{algpseudocode}
\usepackage{pifont}
\usepackage{varwidth}

\input{newcommands.tex}

\newtheorem{problem}{{\bf Problem}}

\begin{document}

\vspace{-20mm}

\title{FDD Massive MIMO via UL/DL Channel Covariance Extrapolation and  Active Channel Sparsification}
\author{Mahdi Barzegar Khalilsarai$^\star$, Saeid Haghighatshoar$^\star$,  Xinping Yi$^\dagger$,  and Giuseppe Caire$^\star$
\thanks{$\star$ Communications and Information Theory Group, Technische Universit\"{a}t Berlin (\{m.barzegarkhalilsarai, saeid.haghighatshoar, caire\}@tu-berlin.de).}
\thanks{$\dagger$ Department of Electrical Engineering and Electronics, University of Liverpool  (xinping.yi@liverpool.ac.uk).}
}

\maketitle
\vspace{-0.5cm}

\begin{abstract}
\linespread{1.3}\small
We propose a novel method for massive Multiple-Input Multiple-Output (massive MIMO) in 
Frequency Division Duplexing (FDD) systems. Due to the large frequency separation between 
Uplink (UL) and Downlink (DL),  in FDD systems channel reciprocity does not hold. Hence, 
in order to provide DL channel state information to the Base Station (BS),
closed-loop DL channel probing and {\em Channel State Information} (CSI) feedback is needed. In massive MIMO this incurs typically  
a large training overhead. For example, in a typical configuration with
$M \simeq 200$ BS antennas and fading coherence block of $T \simeq 200$ symbols, 
the resulting rate penalty factor due to the DL training overhead, given by  $\max\{0, 1 - M/T\}$, 
is close to 0. To reduce this overhead, we build upon the well-known fact 
that the {\em Angular Scattering Function} (ASF)  of the user channels  is {\em invariant} over frequency intervals
whose size is small with respect to the carrier frequency (as in current FDD cellular standards). 
This allows to estimate the users' DL channel covariance matrix from UL pilots without additional overhead.  
Based on this covariance information, we propose a novel {\em sparsifying precoder} in order to
maximize the rank of the effective sparsified channel matrix subject to the condition that each effective user channel has 
sparsity not larger than some desired DL pilot dimension  $\Tdl$, resulting in the DL training overhead factor $\max\{0, 1 - \Tdl / T\}$ and
CSI feedback cost of  $\Tdl$ pilot measurements.
The optimization of the sparsifying precoder is formulated as a {\em Mixed Integer Linear Program}, 
that can be efficiently solved. Extensive simulation results demonstrate the superiority of the proposed approach 
with respect to concurrent state-of-the-art schemes based on compressed sensing or UL/DL dictionary learning.
\end{abstract}

\vspace{-0.5cm}

\begin{IEEEkeywords}
FDD massive MIMO,  downlink covariance estimation, active channel sparsification.  
\end{IEEEkeywords}


\section{Introduction}  \label{sec:intro}

Multiuser \textit{Multiple-Input Multiple-Output} (MIMO) consists of exploiting multiple antennas at the Base Station (BS) side, in order to
multiplex over the spatial domain multiple data streams to multiple users sharing the same time-frequency transmission resource (channel bandwidth and time slots). 
For a block-fading channel with spatially independent fading and coherence block of $T$ symbols,\footnote{This is the number of signal dimensions over which
the fading channel coefficients can be considered constant over time and frequency \cite{tse2005fundamentals}.}
the high-SNR sum-capacity behaves as $C(\SNR) = M^* (1  - M^*/T) \log \SNR + O(1)$, 
where $M^* = \min\{M,K,T/2\}$, $M$ denotes the number of BS antennas, 
and $K$ denotes the number of single-antenna users \cite{zheng2002communication,marzetta2006much,adhikary2013joint}. 
When $M$ and the number of users are potentially very large, 
the system {\em pre-log factor}\footnote{ With this term we indicate the the number of spatial-domain data streams supported by the system,  
such that each stream has spectral efficiency that behaves as an interference-free Gaussian channel, i.e., 
$\log \SNR + O(1)$. In practice, although the system may be interference limited (e.g., due to inter-cell interference in multicell cellular systems), 
a well-design system would exhibit a regime of practically relevant SNR for which its sum-rate behaves as an affine function of $\log \SNR$ 
\cite{lozano2013fundamental}.}
is maximized by serving $K = T/2$ data streams (users). While any number $M \geq K$ of BS antennas yields the same (optimal) 
{\em pre-log factor}, a key observation made in \cite{Marzetta-TWC10} is that,  when training a very large number of 
antennas comes at no additional overhead cost, it is indeed convenient to use $M \gg K$ antennas at the BS. 
In this way, at the cost of some additional {\em hardware complexity}, 
very significant benefits at the system level can be achieved. These include:  
i) energy efficiency (due to the large beamforming gain); ii) 
inter-cell interference reduction; iii) a dramatic simplification of user scheduling and rate adaptation,  
due to the inherent large-dimensional channel hardening \cite{larsson2014massive}. 
Systems for which the number of BS antennas $M$ is much larger than the number of DL data streams $K$ 
are generally referred to as {\em massive MIMO} (see \cite{Marzetta-TWC10,larsson2014massive,marzetta2016fundamentals} and references therein). 
Massive MIMO  has been the object of intense research investigation and development
and is expected to be a cornerstone of the forthcoming 5th generation of wireless/cellular 
systems \cite{boccardi2014five}.   

In order to achieve the benefits of massive MIMO,  the BS must learn 
the downlink channel coefficients for $K$ users and $M \gg K$ BS antennas. 
For Time Division Duplexing (TDD) systems, due to the inherent \textit{Uplink}-\textit{Downlink} (UL-DL) channel reciprocity \cite{marzetta2006much}, this can be 
obtained from $K$ mutually orthogonal  UL pilots transmitted by the users. 
Unfortunately, the UL-DL channel reciprocity does not hold for \textit{Frequency Division Duplexing} (FDD) systems, 
since the  UL and DL channels are separated in frequency by much more than the channel coherence bandwidth \cite{tse2005fundamentals}. 
Hence, unlike TDD systems, in FDD the BS must actively probe the DL channel by sending a common DL pilot signal, and request
the users to feed their channel state back. 

In order to obtain a ``fresh'' channel estimate for each coherence block,  
$\Tdl$ out of $T$ symbols per coherence block  must be dedicated to the DL common pilot. 
Assuming (for simplicity of exposition) a delay-free channel state feedback, 
the resulting DL {\em pre-log factor} is given by $K \times \max\{0, 1 - \Tdl/T\}$, where 
$K$ is the number of served users, and $\max\{0, 1 - \Tdl/T\}$ is the penalty factor incurred by  DL channel training. 
Conventional DL training consists of sending orthogonal pilot signals from each BS antenna. Thus, in order to train $M$ antennas, 
the minimum required training dimension is $\Tdl = M$. 
Hence, with such scheme, the number of BS antennas $M$ cannot be made arbitrarily large.  For example, consider a typical case taken from the 
LTE system \cite{sesia2011lte}, where groups of users are scheduled over 
resource blocks spanning 14 OFDM symbols $\times$ 12 subcarriers, for a total dimension of $T =  168$ symbols in the time-frequency plane. 
Consider a typical massive MIMO configuration serving $K \sim 20$ users with $M \geq 200$ antennas (e.g., see \cite{malkowsky2017world}). 
In this case, the entire resource block dimension would be consumed by the DL pilot, leaving no room for data communication.  
Furthermore, feeding back the $M$-dimensional measurements (or estimated/quantized channel vectors) 
represents also a significant feedback overhead for the UL \cite{caire2010multiuser,kobayashi2011training,yin2013coordinated,love2003grassmannian,jindal2006mimo}. 

While the argument above is kept informal on purpose, it can be made information-theoretically rigorous. 
The central issue is that, if one insists to estimate the $K \times M$ channel matrix in an ``agnostic'' way, i.e., without exploiting 
the channel fine structure, a hard dimensionality bottleneck kicks-in and fundamentally limits the number of data streams that can be supported in the DL by 
FDD systems. It follows that gathering ``massive MIMO gains'' in FDD systems is a challenging problem. 
On the other hand, current wireless networks are mostly based on FDD. Such systems are easier to operate and
more effective than TDD systems in situations with symmetric traffic and delay-sensitive applications 
\cite{jiang2015achievable,chan2006evolution,rao2014distributed}. In addition, converting current FDD systems to TDD
would represent a non-trivial cost for wireless operators.  With these motivations in mind, a significant effort has been 
recently devoted in order to reduce the common DL training dimension and feedback overhead 
in order to materialize significant massive MIMO gains also for FDD systems. 

\subsection{Related works: compressed DL pilots} \label{cs-based-works}
Several works have proposed to reduce both the DL training and UL feedback overheads by exploiting the sparse structure of the 
massive MIMO channel. In particular, these works assume that propagation between the BS array and the user antenna occurs through a limited number
of scattering clusters, with limited support\footnote{Throughout the paper the term ``support" indicates a set of intervals/indices over which a function/vector has non-zero value.} in the {\em Angle-of-Arrival/Angle-of-Departure} (AoA-AoD) domain.\footnote{From the BS perspective, AoD for the DL and AoA for the UL indicate the same domain. Hence, 
we shall simply refer to this as the ``angle domain'', while the meaning of departure (DL) or arrival (UL) is clear from the context.}  
Hence, by decomposing the angle domain into discrete ``virtual beam'' directions, 
the $M$-dimensional user channel vectors admit a sparse representation in the beam-space 
domain (e.g., see \cite{sayeed2002deconstructing,bajwa2010compressed}). 
Building on this idea, a large number of works (e.g., see \cite{kuo2012compressive,sim2016compressed,rao2014distributed,gao2015spatially,ding2016dictionary,fang2017low,dai2017fdd,xie2017unified}) proposed to use  ``compressed pilots'', i.e., a  reduced DL pilot dimension $\Tdl < M$, in order to 
estimate the channel vectors using {\em Compressed Sensing} (CS) techniques 
\cite{donoho2006compressed,candes2008introduction}. In \cite{bajwa2010compressed} sparse representation of channel multipath components  in angle, delay and Doppler domains was exploited to propose CS methods for channel estimation using far fewer measurements than required by conventional least-squares (LS) methods. For example, in \cite{gao2015spatially}, the authors noticed that the angles of the multipath channel components are common among all the subcarriers in the OFDM signaling and exploited the common sparsity to further reduce the number of required channel measurements. This gives rise to a so-called Multiple Measurement Vector (MMV) setting, arising when multiple snapshot of a random vector with common sparse support can be acquired and jointly processed (e.g., see \cite{chen2006theoretical,eldar2010average}). This was adapted to FDD in massive MIMO regime were introduced next, where the frequent idea is to probe the channel using compressed pilots in the downlink, receiving the measurements at the BS via feedback and performing channel estimation there. A recent work based on this approach was presented in \cite{rao2014distributed}, starting with the observation that, as shown in many
	experimental studies \cite{kyritsi2003correlation,kaltenberger2008correlation,hoydis2012channel,gao2011linear},
	the propagation between the BS antenna array and the users occurs along given scattering clusters, 
	that may be common to multiple users, since they all belong to the same scattering environment.
	In turns, this yields that the channel sparse representations (in the angle/beam-space domain) share a common part of their support. 
	Hence, \cite{rao2014distributed} considers a scheme where the users feed back their noisy DL pilot measurements to the BS 
	and the latter runs a {\em joint recovery} algorithm, coined as \textit{Joint Orthogonal Matching Pursuit} (J-OMP), able to take advantage of the 
	common sparsity.  It follows that in the presence of common sparsity, J-OMP  improves upon the basic CS schemes that estimate 
	each user channel separately.

More recent CS-based methods, in addition, make use of the \textit{angular reciprocity} between the UL and the DL channels in FDD systems 
to improve channel estimation. Namely, this refers to the fact that the directions (angles) of propagation for the UL and DL channel are invariant
over the frequency range spanning the UL and DL bands, which is generally very small with respect to the carrier frequency (e.g., UL/DL separation of the order of
100MHz, for carrier frequencies ranging between 2 and 6 GHz)  \cite{hugl2002spatial,ali2017millimeter,xie2017channel}. 
In \cite{xie2017unified} the sparse set of AoAs is estimated from a preamble transmission phase in the UL, and this information
is used for user grouping and channel estimation in the DL according to the well-known JSDM paradigm \cite{adhikary2013joint,nam2014joint}. 
In \cite{ding2016dictionary} the authors proposed a dictionary learning-based approach. First,  in a preliminary learning 
phase a pair of UL-DL dictionaries able to sparsely representing the channel are obtained. Then, these dictionaries are used for a joint sparse estimation 
of instantaneous UL-DL channels. An issue with this method is that the dictionary learning phase requires off-line training and must be re-run if the propagation environment around
the BS changes (e.g., due to large moving objects such as truck and buses, or new building). 
In addition, the computation involved in the instantaneous channel estimation is prohibitively demanding for real-time operations with a large 
number of antennas ($M>100$).  
In \cite{dai2017fdd} the authors propose estimating the DL channel using a sparse Bayesian learning framework aiming at joint maximum a posteriori (MAP) estimation 
of the off-grid AoAs and multipath component strength by observing instantaneous UL channel measurements. 
This method has the drawback that it fundamentally assumes discrete and separable (in the AoA domain) multipath components and assumes that the order of the channel
(number of AoA components) is a priori known. Hence, the method simply cannot be applied in the case of continuous (diffuse) scattering, 
where the scattering power is distributed over a continuous interval of in the angle domain.
	\subsection{Contribution}\label{sec:contribution}
	The focus of this paper is an efficient scheme for massive MIMO in FDD systems. 
	Our goal is to be able to serve as many users as possible even with very small number of DL pilots, compared to the inherent channel dimension. Similar to previous works \cite{rao2014distributed,ding2016dictionary,dai2017fdd}, 
	we consider a scheme where each user sends back its $\Tdl$ noisy pilot observations per slot, using unquantized analog feedback  (see \cite{caire2010multiuser,kobayashi2011training}). Hence, achieving a small $\Tdl$ yields  both a reduction of DL training and UL feedback overhead. We summarize the major contributions of our work as follows:\\
	$\bullet$ \textbf{DL covariance estimation:} the first problem addressed in this paper is how to estimate DL channel covariance from UL pilot symbols, which are sent anyway in order to enable a coherent multiuser MIMO reception in the UL (see Section \ref{sec:DL_cov_est}). The  covariance matrix can be expressed as an integral transform of the channel \textit{Angular Scattering Function} (ASF), which encodes the signal power distribution over the angle domain. 
	Because of the already mentioned UL/DL angle reciprocity, the channel ASF is invariant with respect to frequency
over frequency intervals that are small with respect to the carrier frequency.
Stemming from the ASF reciprocity, the idea of UL to DL covariance estimation/transformation is studied in several previous works, including \cite{chalise2004robust,han2010potential,aste1998downlink,vasisht2016eliminating,miretti2018fdd}. 
Our approach consists of estimating the channel ASF of each user from UL pilots, and using it to ``extrapolate'' the covariance matrix from UL to DL. 
As shown in our recent work \cite{haghighatshoar2018multi},  this extrapolation problem is non-trivial and must be posed in a robust min-max sense. In \cite{haghighatshoar2018multi} we also show that robust covariance reconstruction can be obtained as long as one ensures that the estimated channel ASF is a real, positive function and that its generated UL antenna correlation is consistent with the true UL antenna correlation. 
Unlike most of the works in the literature, including the ones mentioned above, our covariance extrapolation technique does not rely on any regularity assumption on the ASF. 
That is to say, we do not assume the ASF to be discrete or sparse, and the estimation method works for a generic ASF. In contrast, it exploits the Toeplitz (resp., block-Toeplitz) structure of  the channel covariance matrix resulting from Uniform Linear Arrays (ULA) (resp., Uniform Planar Arrays (UPA)). \\
	$\bullet$ \textbf{Active channel sparsification:} the second problem addressed in this paper is how to effectively and {\em artificially} reduce each user channel 
	dimension, such that a single common DL pilot of assigned dimension $\Tdl$ is sufficient to estimate a large number of user channels (see Section \ref{sec:sparification}).  
	In the CS-based works reviewed above, the pilot dimension depends on the channel sparsity level $s$ (number of non-zero components in the 
	angle/beam-space domain). In fact, standard CS theory states that stable sparse signal reconstruction is possible using 
	$\Tdl = O(s \log M)$ measurements.\footnote{As commonly defined in the CS literature, 
		we say that a reconstruction method is stable if the resulting MSE vanishes as $1/\SNR$, where  $\SNR$ denotes the Signal-to-Noise Ratio of the measurements.} 
	In a rich scattering situation, $s$ is large or may in fact vary from user to user or in different cell locations. Even if the channel support is known, one needs at least $s$ measurements for a stable channel estimation.
	Hence, these CS-based methods (including the ones having access to support information) may or may not work well, depending on the propagation environment. 
	In order to allow channel estimation with an {\em assigned} pilot dimension $\Tdl$,  we use the DL covariance information 
	in order to design an optimal {\em sparsifying precoder}. This is a linear transformation that depends only
	on the channel second order statistics (estimated DL covariances) that imposes that the effective channel matrix  (including the precoder) 
	has large rank and yet each column has sparsity not larger than $\Tdl$. In this way, our method is not at the mercy of nature, i.e. it is flexible with respect to various types of environments and channel sparsity orders.
	We cast the optimization of the sparsifying precoder as a Mixed Integer Linear Program (MILP),  which can be efficiently solved using standard 
	off-the-shelf solvers.
\section{System Setup}\label{sec:sys_setup}

We consider a directional channel propagation model formed by multiple multipath components (MPCs), each corresponding to 
a scattering cluster characterized by a certain angle width and AoA direction. 
In addition, as in  \cite{rao2014distributed}, we consider the possibility that different users have partially overlapped
multipath components. An example of such spatially consistent scattering model is provided by the COST 2100 channel model \cite{liu2012cost}, where 
each MPC is associated to a visibility region, and users inside its visibility region are coupled with the BS array through the corresponding scattering cluster 
(see Fig. \ref{COST2100}). 

\begin{figure}[t]
	\centering
	\includegraphics[width=8cm]{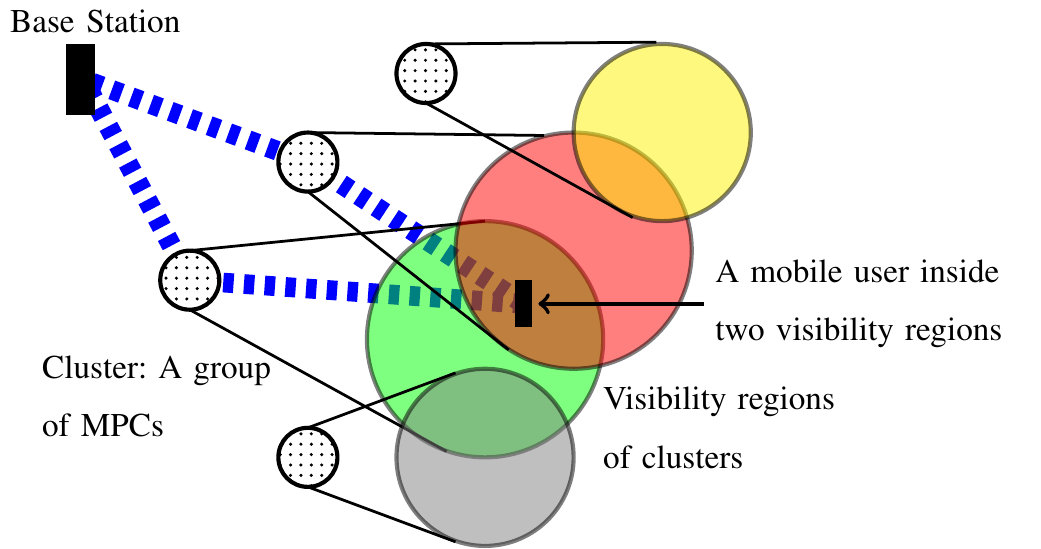}
	\caption{\small A sketch of the clusters and visibility regions in the COST 2100 model.}
	\label{COST2100}
\end{figure}

This model implies that the scattering geometry of the channel between the BS antenna array and the UE antenna 
remains constant over time intervals corresponding to the UE remaining in the same intersection of visibility regions. 
Since moving across the regions occurs at a time scale much larger than moving across one wavelength, 
it is safe to assume that the channel scattering geometry is locally stationary over intervals much longer than 
the time scale of the transmission of channel codewords. 
Such fixed geometry yields the so-called \textit{Wide Sense Stationary Uncorrelated Scattering} (WSSUS) channel model, for which the channel
vectors evolve in time according to a WSS processes. Also, we use the ubiquitous  block-fading approximation, and assume that the channel random 
process can be approximated as locally piecewise constant over blocks of 
$T$ time-frequency symbols, where $T \approx W_c T_c$, $W_c$ denoting the channel coherence bandwidth and $T_c$ denoting the channel coherence time \cite{tse2005fundamentals}.
We consider a BS equipped with an ULA with $M \gg 1$ antennas and single-antenna UEs.\footnote{ The approach of this paper can be immediately generalized to
UPAs for 3-dim beamforming. Here we restrict to a planar geometry for the sake of simplicity.} 
In an FDD system, communication takes place over two disjoint frequency bands. 
The UEs transmit to the BS over the frequency interval $[\ful-\frac{W_{\rm ul}}{2},\ful+\frac{W_{\rm ul}}{2}]$, 
where $\ful$ is the UL carrier frequency and $W_{\rm ul}$ is the UL bandwidth. 
Likewise, the BS transmits to the UEs over the frequency band $[\fdl-\frac{W_{\rm dl}}{2},\fdl+\frac{W_{\rm dl}}{2}]$ where $\fdl$ is the DL carrier frequency 
and $W_{\rm dl}$ is the DL bandwidth. 
The channel bandwidth is always much less than the carrier frequency, i.e. $\frac{W_{\rm ul}}{\ful}\ll 1$,  $\frac{W_{\rm dl}}{\fdl}\ll 1$. 
Let $\alpha = \frac{\fdl}{\ful}$ denote the ratio between the DL and the UL carrier frequencies. Notice that in FDD systems in operation 
today, we always have $\alpha > 1$ (e.g., see \cite{rel14}). A general form for the above WSSUS channel model in the time-frequency-antenna domain is given by
\begin{equation}\label{eq:ch_vec}
\hv (t,f)= \int_{\Theta} \rho(t,d\theta) \av(\theta,f) \in \bC^M,
\end{equation}
where $\Theta:=[-\theta_{\max}, \theta_{\max})$ is the angular range scanned by the ULA, the vector $\av(\theta,f) \in \bC^M$ 
is the {\em array response} at frequency $f$ and angle $\theta$, with $m$-th element given by
	\begin{equation}\label{eq:a_vec}
		[\av(\theta,f) ]_m= e^{j 2\pi \frac{f}{c_0} m d \sin\theta},
	\end{equation} 
	where $c_0$ denotes the speed of light and $d$ the distance between two consecutive antennas, and  $\rho(t,d\theta)$ is a random gain dependent on the time  $t$ and the angle range $[\theta,\theta + d\theta]$. 
	We model
	$\rho(t,d\theta)$ to be a {\em zero-mean} Gaussian stochastic process with independent increments respect to $\theta$ (uncorrelated scattering) and WSS with respect to $t$. The angular autocorrelation function is given by
	\begin{equation}\label{eq:autocorr_of_ro}
	\bE \left[\rho(t,d\theta) \rho(t,d\theta')\right] = \gamma (d\theta) \delta (\theta - \theta'),
	\end{equation}
	where $\gamma (d\theta)$ is the channel ASF,  modeling the power received from scatterers located at any angular interval. 
	It is convenient to assume that $\gamma (d\theta)$ is a normalized density function, such that $\int_{\Theta} \gamma (d\theta) = 1$. 
	Based on the narrow-band assumption we consider the array response to be a constant function of frequency over each of the UL and DL bands 
	separately and write $\av_{\rm ul}(\theta) := \av(\theta,\ful)$ and $\av_{\rm dl}(\theta) := \av(\theta,\fdl)$. 
	We let $d = \kappa \frac{\lambdaul}{2\sin (\theta_{\max})}$, where $\lambdaul=\frac{\ful}{c_0}$ is the UL carrier wavelength and $\kappa$ 
	is the spatial oversampling factor, usually (including here) set to $\kappa = 1$. With this definition we have that $[\av_{\rm ul}(\theta) ]_m = e^{j m \pi \frac{\sin (\theta)}{\sin (\theta_{\max})}}$ and $[\av_{\rm dl}(\theta) ]_m = e^{j m \pi \alpha \frac{\sin (\theta)}{\sin (\theta_{\max})}}$. 
	Notice that the exponents of the array response elements for UL and DL differ by the factor $\alpha$, which is typically slightly larger than 1 
	(e.g., for the LTE-IMT bands  we have $\alpha = \frac{2140}{1950} \approx 1.1$ \cite{rel14}).
	
The channel vector covariance matrix is thereby given as follows
\begin{equation}\label{eq:cov_mat}
	\Cm_{\hv} (f) = \bE \left[ \hv (t,f) \hv (t,f)^\herm \right] = \int_{\Theta} \gamma (d\theta) \av (\theta,f) \av (\theta,f)^\herm, 
\end{equation}
which is time-invariant due to stationarity. The dependence of the covariance matrix on frequency is due to the fact that, as discussed before, 
the array response vector is a function of frequency.  The covariance matrix is Toeplitz positive semidefinite Hermitian and hence can be described by its first column $\cv (f)$ as $\Cm_{\hv} (f) = \Tc \left(\cv (f)\right)$,\footnote{
	For $\xv \in \bC^M$, we let $\Tc (\xv)$ denote the Toeplitz Hermitian matrix with first column $\xv$, i.e., with $(i,j)$-th element
	$[\Tc(\xv)]_{i,j} = x_{i-j}$ for $i \geq j$ and $[\Tc(\xv)]_{i,j} = x^*_{|i-j|}$ for $i < j$. If $\xv$ is a sampled autocorrelation function, 
	then $\Tc (\xv)$ is positive semidefinite.}  
where the first column is given by $\cv (f) = \int_{\Theta} \gamma (d\theta) \av (\theta,f)$. 
We denote UL and DL covariance matrices by $\Cul:=\Cm_{\hv} (\ful)$ and $\Cdl:=\Cm_{\hv} (\fdl)$, respectively. 
\section{DL Covariance Estimation from UL Pilots}\label{sec:DL_cov_est}
Our proposed DL covariance estimation method exploits the assumption that the channel ASF 
is the same for UL and DL (angular reciprocity) \cite{hugl2002spatial,ali2017millimeter,xie2017channel}. Unlike previous works, we do not assume the ASF to be sparse, or to consist of only ``discrete" components. In fact, as we have shown in a companion paper \cite{haghighatshoar2018multi}, any estimate of the ASF that is real, positive and consistent with the UL covariance, regardless of being sparse, is good enough for the purpose of DL covariance estimation.
\subsection{Uplink covariance estimation}
Since the user channel vectors are mutually independent, 
and we assume Additive White Gaussian Noise (AWGN), the estimation of each user channel covariance in the UL is decoupled and we can focus on the estimation of 
a generic user. The received UL pilot observation during the $i$-th UL coherence block, after projecting over the orthogonal pilot
sequence of the given generic user, is given by $\yv[i] = \bfh_{\rm ul}[i] + \bfn[i]$ (see \cite{Marzetta-TWC10}),
where $\bfh_{\rm ul}[i]$ denotes the generic user channel vector during the $i$-th coherence block
and where $\bfn \sim \cg({\bf 0}, \sigma^2 \Id_M)$ is the measurement noise vector. 
Collecting a window of $N_{\rm ul}$ UL measurements and assuming the noise variance $\sigma^2$ to be known we estimate the UL covariance as follows. We first calculate the sample covariance matrix as $ \tilde{\Cm}_{\rm ul} =  \frac{1}{N_{\rm ul} } \sum_{i=1}^{N_{\rm ul}} \yv[i] \yv[i]^\herm$. The sample covariance is not necessarily Toeplitz and therefore, to improve the estimate, we project it to the Toeplitz, positive semidefinite cone using the following convex program as suggested in \cite{miretti2018fdd},  
\begin{equation}\label{eq:ul_samp_cov}
 \hat{\Cm}_{\rm ul} = \underset{\Xm \in \Tm_+^M}{\arg\min} ~ \Vert \Xm -  \left( \tilde{\Cm}_{\rm ul} - \sigma^2 \mathbf{I}_M \right)\Vert_F, 
\end{equation}
 where $\Tm_+^M$ is the cone of Toeplitz, Hermitian, positive semidefinite $M\times M$ matrices and $\Vert \cdot \Vert_F$ is the Frobenius norm. Being a Toeplitz Hermitian matrix, $ \hat{\Cm}_{\rm ul}$ can be fully described by its first column which is denoted by $\hat{\cv}_{\rm ul}$ hereafter.
\subsection{Estimation of the channel ASF}

Define $\Gc$ as a uniform grid consisting of $G\gg M$ discrete angular points $\{\theta_i\}_{i=1}^G$, 
where each point is given by $\theta_i=\sin ^{-1} \left((-1 + \frac{2(i-1)}{G}) \sin (\theta_{\max})\right) \in \Theta$, 
and define $\Gm \in \bC^{M\times G}$ to be a matrix whose $i$\textsuperscript{th} column is given by $\frac{1}{\sqrt{M}}\av_{\rm ul} (\theta_i),~ i\in [G]$. A discrete approximation of the ASF $\gamma$ on the grid $\Gc$ can be written as $\gamma (\rmd \theta)\approx \sum_{i=1}^G [\zv]_i \delta (\theta - \theta_i) $ for some vector $\zv \in \bR^{G}_+$. We find $\zv$ by solving a non-negative least squares (NNLS) convex optimization program \cite{haghighatshoar2018multi}:
\begin{equation}\label{eq:conv_problem_1}
\zv^\ast = \underset{\zv \in \bR_{+}^{G}}{\arg\min} ~ \Vert \Gm \zv - \hat{\cv}_{\rm ul}\Vert.
\end{equation}
The particularly desirable property of NNLS is that, it yields a real, positive approximation of the ASF and, minimizes the $\ell_2$ distance of its generated UL covariance samples $\Gm \zv$ and the estimated UL covariance samples $\hat{\cv}_{\rm ul}$ to satisfy a data consistency constraint. In fact, as we show in \cite{haghighatshoar2018multi}, positivity and data consistency are the only two requirements needed for guaranteeing a stable DL covariance estimation. Furthermore, 
the NNLS solution can be efficiently computed via several convex optimization techniques \cite{bertsekas2015convex}. 
By solving \eqref{eq:conv_problem_1}, the estimated discretized approximation of the ASF is simply given as $\hat{\gamma} (d\theta)= \sum_{i=1}^G [\zv^\ast]_i \delta (\theta - \theta_i) $.

\subsection{Covariance extrapolation via Fourier transform resampling}

Building on the theory developed in our companion paper \cite{haghighatshoar2018multi}, 
the problem of extrapolating the estimated UL covariance matrix to the DL frequency can be seen as 
the resampling of the Fourier transform of the channel ASF. 
To see this, notice that the $m$-th components of the first column $\cv_{\rm ul}$ of $\Cm_{\rm ul}$ are given by  
\begin{equation} \label{gamma-sampled}
[\cv_{\rm ul}]_m= \int_{\Theta} \gamma (d\theta) e^{j m \pi \frac{\sin \theta}{\sin \theta_{\max}}}=\int_{-1}^1 \gamma (d\xi) e^{jm \pi \xi}, ~ m\in[M], 
\end{equation}
where we introduce the change of variable $\xi = \frac{\sin \theta}{\sin \theta_{\max}}$. 
Define the continuous Fourier transform of the positive measure $\gamma (d\xi)$ as $\check{\gamma} (x) = \int_{-1}^1 \gamma (d\xi) e^{j x \pi \xi}$. 
Then it is clear from (\ref{gamma-sampled}) that $[\cv_{\rm ul}]_m = \check{\gamma} (m),~ m\in [M]$. 
In words, the first column of the UL covariance matrix is simply a sampling of the Fourier transform of the positive measure $\gamma (d\xi)$ at points $m = 0 , \ldots, M-1$. 
Taking similar steps, one can show that the components of the first column of the DL covariance matrix are given by $[\cv_{\rm dl}]_m = \int_{-1}^1 \gamma (d\xi) e^{j\alpha m \pi  \xi}, ~ m\in[M] $ and hence $[\cv_{\rm dl}]_m  = \check{\gamma}(\alpha m),~ m\in [M]$. 
Estimating the DL covariance from the UL covariance is equivalent to resampling $\check{\gamma}(\cdot)$ 
over a grid $\{0,\alpha,2\alpha,\ldots,(M-1)\alpha\}$, knowing its samples at the integer grid 
$\{0,1,2,\ldots,M-1\}$. 
Summarizing, the proposed DL covariance estimation method consists of the following steps:
 	1) Estimate a discrete approximation of the positive measure $\gamma (d\theta)$ using the 
	the UL sample covariance estimator and solving \eqref{eq:conv_problem_1}. 
	The samples of the Fourier transform of this measure on the grid $\{0,\ldots,M-1\}$ asymptotically converge to those generated 
	from the true angular scattering function \cite{haghighatshoar2015channel} for large sample size $N_{\rm ul}$.\\
2) Calculate the Fourier transform of the estimated measure on the grid $\{\alpha m \}_{m=0}^{M-1}$ to obtain the estimated DL antenna autocorrelation function
\begin{equation}\label{eq:DL_first_col}
[\hat{\cv}_{\rm dl}]_m = \sum_{i=1}^{G} \hat{\gamma} (\theta_i) e^{j\alpha (m - 1) \pi \frac{\sin \theta_i }{\sin \theta_{\max}}}, ~ m \in [M].
\end{equation}
The resulting DL covariance matrix is given by the Toeplitz completion 
$ \hat{\Cm}_{\rm dl}  =\Tc (\hat{\cv}_{\rm dl})$.
As a final remark in this section, notice that the above DL covariance estimation method does not rely on particular features of the channel ASF. 
For example, it does not require that the ASF has a sparse or discrete support, as needed in other ad-hoc methods
(e.g., see \cite{vasisht2016eliminating,xie2017channel,xie2017unified}). 

\subsection{Circulant approximation of the DL covariance matrices}  \label{sec:circulant}

 The DL covariance estimation from UL pilot signals is performed for all the users $k\in [K]$ at the BS. These covariance matrices are Toeplitz 
 by construction, 
 due to the structure of the ULA as described before. 
 In Section \ref{sec:sparification} we will introduce the novel idea of active channel sparsification where, for a given DL pilot dimension, 
 the BS selects a set of angular directions to transmit data to the users, such that the number of DL data streams 
 that the system can support is maximized. 
 A necessary step before performing sparsification is that all of the estimated DL covariance matrices 
 share a common set of eigenvectors, namely, the same virtual beam-space representation. 
 In the massive MIMO regime where $M \gg 1$, this is possible by considering the circulant approximation of Toeplitz matrices  that follows as an application of Szeg\"o
 Theorem (see details in \cite{adhikary2013joint} and references therein). 
 Let $\Cdlk$ denote the estimated DL channel covariance 
 of user $k$ for $k \in [K]$, where from now on we shall drop the subscript $``{\rm dl}"$ since  it is clear from the context, as we consider 
 only DL multiuser MIMO transmission. Define the diagonal matrices $\mathring{\Lambdam}_k,~ k \in [K]$ for which $[\mathring{\Lambdam}_k]_{m,m} = [\Fm^\herm \Cm_k \Fm]_{m,m}$, where $\Fm$ is the $M\times M$ DFT matrix, whose $(m,n)$-th entry is given by $[\Fm]_{m,n}=\frac{1}{\sqrt{M}} e^{-j2\pi \frac{mn}{M}},~m,n\in [M]$. 
 There are several ways to define a circulant approximation \cite{zhu2017asymptotic}, among which we choose the following:  
 \begin{equation}\label{eq:circ_approx_2}
\Cdlcirck= \Fm \mathring{\Lambdam}_k\Fm^\herm.
 \end{equation}
According to Szeg\"o's theorem, for large $M$, $\mathring{\Lambdam}_k$ converges to the diagonal eigenvalue matrix $\Lambdam_k$ of $\Cm_k$, i.e. $\mathring{\Lambdam}_k \rightarrow \Lambdam_k$ as $M\rightarrow \infty$. 
Hence, within a small error for large $M$, the sought set of (approximate) common eigenvectors for all the users 
is provided by the columns of the $M \times M$ DFT matrix. 
As a consequence, the DL channel covariance of user $k$ is characterized simply 
via a vector of eigenvalues $\lambdav^{(k)} \in \bR^M$, with $m$-th element 
$[\lambdav^{(k)}]_m = [\mathring{\Lambdam}^{(k)}]_{m,m}$. 
In addition, the DFT matrix forms a unitary basis 
for (approximately) expressing any user channel vector via an (approximated) Karhunen-Loeve expansion. In particular, 
let $\fv_m:=[\Fm]_{\cdot,m}$ denote the $m$-th column of $\Fm$. We can express the DL channel vector of user $k$ 
as 
\begin{equation} \label{approximate-KL}
\hv^{(k)} \approx \sum_{m=0}^{M-1} g_{m}^{(k)} \sqrt{[\lambdav^{(k)}]_m} \,\fv_m, 
\end{equation}
where $g_{m}^{(k)} \sim \cg (0,1) $ are i.i.d. random variables. The columns of $\Fm$ are very similar to array response vectors and in fact, recalling equation \eqref{eq:a_vec}, we have that $\fv_m = \frac{1}{\sqrt{M}} \av_{\rm dl} \left(\sin^{-1} (\frac{\lambdadl}{d} \frac{m}{M})\right)$. Hence, each column with index $m\in [M]$ of the DFT matrix can be seen as the array response to an angular direction and 
$[\lambdav^{(k)}]_m$ can be seen as the power of the channel vector associated with user $k$ along that direction. Due to the limited number of local scatterers as seen at the BS and the large number of antennas of the array, only a few entries of $\lambdav^{(k)}$ are significantly large, implying that the DL channel vector $\hv^{(k)}$ is sparse in the Fourier basis.  This sparsity in the beam-space domain is precisely what has been exploited in the CS-based works discussed in Section \ref{cs-based-works}, 
in order to reduce the DL pilot dimension $\Tdl$. It is also evident that this channel representation combined with the geometrically consistent model
reviewed in Section \ref{sec:sys_setup} yields the common sparsity across users, as exploited by J-OMP in \cite{rao2014distributed}. As seen 
in the next section, our proposed approach does not rely on any intrinsic channel sparsity assumption, but adopts a novel artificial sparsification technique.

\section{Active Channel Sparsification and DL Channel Probing}  \label{sec:sparification}

In this section we consider the estimation of the {\em instantaneous realization} of the DL user channel vectors. 
As in \cite{adhikary2013joint}, we consider the concatenation of the physical channel with a fixed precoder, i.e., a linear transformation that may depends on the 
user channel statistics (notably, on their covariance matrices estimated as explained in Section \ref{sec:DL_cov_est}), 
but is independent of the instantaneous channel realizations, which in fact must be estimated via 
the closed-loop DL probing and channel state feedback mechanism as discussed in Section \ref{sec:intro}. 

The BS transmits a training space-time matrix $\Psim$ of dimension $\Tdl \times M'$, such that each 
row $\Psim_{i,.}$ is transmitted simultaneously from the $M' \leq M$ inputs of a precoding matrix $\Bm$ of dimension $M' \times M$, and where
$M'$ is a suitable intermediate dimension that will be determined later. 
The precoded DL training length (in time-frequency symbols) spans 
therefore $\Tdl$ dimensions, and the DL training phase is repeated at each DL slot of dimension $T$. 
Stacking the $\Tdl$ DL training symbols in a column vector, the corresponding observation at the UE $k$ receiver is given by 
\begin{equation}\label{eq:cs_eq_1}
{\bfy}^{(k)}= {\bf\Psi} \bfB {\bfh}^{(k)} + {\bfn}^{(k)} = {\bf\Psi} \heffk + {\bfn}^{(k)},
\end{equation}
where $\Bm$ is the precoding matrix, $\hv^{(k)}$ is the channel vector of user $k$, and we define
$\heffk := \Bm \hv^{(k)}$ as the effective channel vector, formed by the concatenation of the actual DL channel (antenna-to-antenna) 
with the precoder $\Bm$.  The measurement noise is AWGN with distribution $\bfn^{(k)} \sim \cg({\bf 0},\Nvar {\bf I}_{\Tdl})$. 
The training matrix and precoding matrix are normalized such that  
\begin{equation} \label{power-training-phase}
 \trace ( \Psim \Bm \Bm^\herm \Psim^\herm ) = \Tdl \Pdl, 
 \end{equation}
 where $\Pdl$ denotes the total BS transmit power and we define the DL SNR as $\SNR = \Pdl/\Nvar$. 
 Notice that most works on channel estimation focus on the estimation of the  actual channels $\{{\bfh^{(k)}}\}$. 
This is recovered in our setting by letting $\Bm = \Id_M$. However, our goal here is to design a {\em sparsifying precoder} $\Bm$ such that each user effective channel has low dimension (in the beam-space representation) 
and yet the collection of effective channels for $k \in [K]$ form a high-rank matrix. 
In this way, each user channel can be estimated using a small pilot overhead $\Tdl$, but the BS is still able to serve
many data streams using spatial multiplexing in the DL (in fact, as many as the rank of the effective matrix). 

\subsection{Necessity and implication of stable channel estimation} \label{sec:stability}

For simplicity of exposition, in this section we assume that the 
channel representation \eqref{approximate-KL} holds exactly and that the eigenvalue vectors $\lambdav^{(k)}$ have support  
$\Sc_k = \{ m : [\lambdav^{(k)}]_m \neq 0 \}$ with sparsity level $s_k = |\Sc_k|$.
We hasten to say that the above are convenient {\em design assumptions}, made in order to obtain a tractable problem, and that 
the  precoder designed according to our simplifying assumption is applied to the actual physical channels. 
Under these assumptions, the following lemma
yields necessary and sufficient conditions of stable estimation of 
the channel vectors ${\bfh^{(k)}}$. 

\begin{lemma}\label{lem:stable_rec}
Consider the sparse Gaussian vector $\hv^{(k)}$ with support set $\Sc_k$ given by the RHS of (\ref{approximate-KL}). 
Let $\widehat{\hv}^{(k)}$ denote any estimator for $\hv^{(k)}$ 
based on the observation\footnote{Note that this coincides with 
(\ref{eq:cs_eq_1}) with $\Bm=\mathbf{I}_M$, i.e., without the sparsifying precoder.} 
$\yv^{(k)} = {\bf\Psi} {\bfh}^{(k)} + {\bfn}^{(k)}$,  
and  let $\Rm_e = \EE[ (\hv^{(k)} - \widehat{\hv}^{(k)}) (\hv^{(k)} - \widehat{\hv}^{(k)})^\herm ]$ denote the corresponding estimation error covariance matrix.  
If  $\Tdl \ge s_k$ there exist pilot matrices $\Psim \in \CC^{\Tdl \times M}$ for which $\lim_{\Nvar \downarrow 0} \trace(\Rm_e) = 0$ for all support 
sets $\Sc_k : |\Sc_k| = s_k$. 
Conversely,  for any support set $\Sc_k : |\Sc_k| = s_k$ any pilot matrix $\Psim \in \CC^{\Tdl \times M}$ 
with  $\Tdl < s_k$ yields $\lim_{\Nvar \downarrow 0} \trace(\Rm_e) > 0$.
	\hfill $\square$
\end{lemma}
\begin{proof}
	See appendix \ref{App:stable_rec_proof}.
	\end{proof}

As a direct consequence of Lemma \ref{lem:stable_rec}, we have that 
any scheme relying on intrinsic channel sparsity cannot yield stable estimation if $\Tdl < s_k$ for some users. Furthermore,  
we need to impose that the effective channel sparsity (after the introduction of the sparsifying precoder $\Bm$) 
is less or equal to the desired DL pilot dimension $\Tdl$.
It is important to note that the requirement of estimation stability is {\em essential} in order to achieve high spectral efficiency in high SNR conditions, irrespectively of the DL precoding scheme. 
In fact,  if the estimation MSE of the user channels does not vanish as $\Nvar \downarrow 0$, the system self-interference
due to the imperfect channel knowledge grows proportionally to the signal power, yielding a Signal-to-Interference plus Noise Ratio (SINR) that
saturates to a constant when SNR becomes large. Hence, for sufficiently high $\SNR$, the best strategy would consist 
of transmitting just a single data stream,  since any form of multiuser precoding would inevitably lead to an interference limited regime, 
where the sum rate remains bounded  while $\SNR \rightarrow \infty$ \cite{davoodi2016aligned}. In contrast, it is also well-known that 
when the channel estimation error vanishes as $O(\Nvar)$ for $\Nvar \downarrow 0$, the high-SNR sum rate 
behaves as if the channel was perfectly known and can be achieved by very simple linear precoding \cite{caire2010multiuser}. 
A possible solution to this problem consists of serving only the users whose channel support $s_k$ is not larger than $\Tdl$. 
This is assumed {\em implicitly} in all CS-based schemes  (see Section \ref{cs-based-works}),  and represents a major 
intrinsic limitation of the CS-based approaches. In contrast, by artificially sparsifying the user channels, we manage to serve 
all users given a fixed DL pilot dimension $\Tdl$.

\subsection{Sparsifying precoder optimization} \label{sec:sparseB}
Before proceeding in this section, we introduce some graph-theoretic terms \cite{diestel2005graph}. A bipartite graph is a graph whose vertices (nodes) can be divided into two sets $\Vc_1$ and $\Vc_2$, such that every edge in the set of graph edges $\Ec$ connects a vertex in $\Vc_1$ to one in $\Vc_2$. One can denote such a graph by $\Lc = (\Vc_1,\Vc_2,\Ec)$. A subgraph of $\Lc$ is a graph $\Lc' = (\Vc_1',\Vc_2',\Ec')$ such that $\Vc_1'\subseteq \Vc_1$, $\Vc_2'\subseteq \Vc_2$ and $\Ec'\subseteq \Ec$. With regards to $\Lc$, the following terms shall be defined and later used. \\
\textbullet ~\textbf{Degree of a vertex:} for a vertex $x\in \Vc_1\cup \Vc_2$, the degree of $x$ refers to the number of edges in $\Ec$ incident to $x$ and  is denoted by $\text{deg}_{\Lc}(x)$.\\
\textbullet ~\textbf{Neighbors of a vertex:}  the neighbors of a vertex $x\in \Vc_1\cup \Vc_2$ are the set of vertices $y\in \Vc_1\cup \Vc_2$ connected to $x$. This set is denoted by $\Nc_{\Lc}(x)$. \\
\textbullet ~\textbf{Matching:} a matching in $\Lc$ is a subset of edges in $\Ec$ without common vertices.\\
\textbullet ~\textbf{Maximal matching:} a maximal matching $\Mc$ of $\Lc$ is a matching with the property that if any edge outside $\Mc$ and in $\Ec$ is added to it, it is no longer a matching.\\
\textbullet ~\textbf{Perfect matching:} a perfect matching in $\Lc$ is a matching that covers all vertices of $\Lc$.

We propose to design the sparsifying precoder using a graphical model, where a bipartite graph is formed by a set of vertices representing users on one side and another set of vertices representing beams on the other side. An edge of the bipartite graph between a beam and a user represents the presence of that beam in the user angular profile, with its weight denoting the user channel power along that beam. Now, we wish to design the precoder $\Bm$ such that the support of the effective channels $\heffk = \Bm \hv^{(k)}$ is not 
larger than $\Tdl$ for all $k$, such that all users have a chance of being served. 
Let $\Hcoeffs = \Lm \odot \bG\in \bC^{M\times K}$ denote the matrix of DL channel coefficients expressed in the DFT basis (\ref{approximate-KL}),  
in which each column of $\Hcoeffs$ represents the coefficients vector of a user,  where $\Lm$ is a $M \times K$ matrix with elements 
$[\Lm]_{m,k} = \sqrt{[\lambdav^{(k)}]_m}$,  where $\bG \in \bC^{M\times K}$ has i.i.d. elements  $[\bG]_{m,k} = g_m^{(k)} \sim \cg (0,1)$, 
and where $\odot$ denotes the Hadamard (elementwise) product. 
Let $\Am$ denote a one-bit thresholded version of $\Lm$, such that $[\Am]_{m,k}=1$ if $[\lambdav^{(k)}]_m > {\sf th}$, where ${\sf th} > 0$ is a suitable
small threshold, used to identify the significant components,  and consider the $M \times K$ bipartite 
graph $\Lc = \left(\Ac,\Kc,\Ec\right)$ with adjacency matrix $\Am$ and weights $w_{m,k} = [\lambdav^{(k)}]_m$  on the edges $(m,k) \in \Ec$. 

Given a pilot dimension $\Tdl$, our goal consists in selecting a subgraph $\Lc'=\left(\Ac',\Kc',\Ec'\right)$ of $\Lc$ in which each node on either side of the graph has a degree at least 1 and such that: 
\begin{enumerate}
	\item For all $k \in \Kc'$ we have $\text{deg}_{\Lc'}(k) \le \Tdl$, where $\text{deg}_{\Lc'}$ denotes the degree of a node in the selected subgraph.
	\item The sum of weights of the edges incident to any node $k \in \Kc'$ in the subgraph $\Lc'$ is greater than a threshold, i.e. $\sum_{m\in \Nc_{\Lc'} (k)} w_{m,k} \ge \Pthresh, ~\forall k\in \Kc'$.
	\item The channel matrix $\Hcoeffs_{\Ac',\Kc'}$ obtained from $\Hcoeffs$ by selecting $a \in \Ac'$ (referred to as ``selected beam 
	directions") and $k\in \Kc'$ (referred to as ``selected users") has large rank.
\end{enumerate}    
The first criterion enables stable estimation of the effective channel of any selected user with only $\Tdl$ common pilot dimensions 
and $\Tdl$ complex symbols of feedback per selected user. 
The second criterion makes sure that the effective channel strength of any selected user is greater than a desired threshold, since we do not want to spend 
resources on probing and serving users with weak effective channels (where ``weak'' is quantitatively determined by the value of $\Pthresh$). 
Therefore $\Pthresh$ is a parameter that serves to obtain a tradeoff between the rank of the effective matrix (which ultimately determines the number of spatially multiplexed DL data streams) and the beamforming gain (i.e., the power effectively conveyed along each selected user effective channel).
The third criterion is motivated by the fact that the DL pre-log factor is given by ${\rm rank} ( \Hcoeffs_{\Ac',\Kc'} ) \times \max \{ 0 , 1 - \Tdl/T\}$, 
and it is obtained by serving a number of users equal to  the rank of the effective channel matrix. 
The following lemmas relate the rank of the effective channel matrix to a graph-theoretic quantity, namely,  the size of the maximal matching.
\begin{lemma}\label{lem:CUR}
	[Skeleton or ``$\Cm \Um \Rm$" decomposition \cite{goreinov1997theory}] Consider $\Hcoeffs \in \bC^{M\times K}$, of rank $r$. Let $\Qm$ be an $r\times r$ non-singular
	intersection submatrix obtained by selecting $r$ rows and $  r$ columns of $\Hcoeffs$. Then, we have $\Hcoeffs = \Cm \Um \Rm$,
where $\Cm \in \bC^{M\times r}$ and $\Rm \in \bC^{r\times K}$ are the matrices of the selected columns and rows forming the intersection $\Qm$ and $\Um=\Qm^{-1}$.
\hfill $\square$
\end{lemma}

\begin{lemma}\label{lem:matching}
	[Rank and perfect matchings] Let $\Qm$ denote an $r\times r$ matrix with some elements identically zero,
	and the non-identically zero elements independently drawn from a continuous distribution. Consider the associated
	bipartite graph with adjacency matrix $\Am$ such that $\Am_{i,j}=1$ if $\Qm_{i,j}$ is not identically zero, and $\Am_{i,j} = 0$ otherwise.
	Then, $\Qm$ has rank $r$ with probability 1 if and only if the associated bipartite graph contains a perfect matching. \hfill $\square$
\end{lemma}
\begin{proof}
	See appendix \ref{app:proof_lemma_matching}.
\end{proof}
		
A similar theorem can be found in \cite{tutte1947factorization}, but we provide a direct proof in Appendix \ref{app:proof_lemma_matching} for the sake of completeness.
Lemmas \ref{lem:CUR} and \ref{lem:matching} result in the following corollary, which is an original contribution of this work.
\begin{corollary}\label{corollary:matching}
	The rank $r$ of a random matrix $\Hcoeffs \in \bC^{M\times K}$ with either identically zero elements or elements independently drawn from a continuous distribution is given, with probability 1, by the size of the largest intersection submatrix whose associated bipartite graph (defined as in Lemma \ref{lem:matching}) contains a perfect matching.   \hfill $\square$
\end{corollary}
Obviously this corollary holds in our case where the non-zero elements of $\Hcoeffs$ are drawn from the complex Gaussian distribution. 
Using Corollary \ref{corollary:matching} this problem can be formulated as:
\begin{figure}[t]
	\centering
	\begin{subfigure}[b]{.35\textwidth}
		\centering
		\includegraphics[width=1\linewidth]{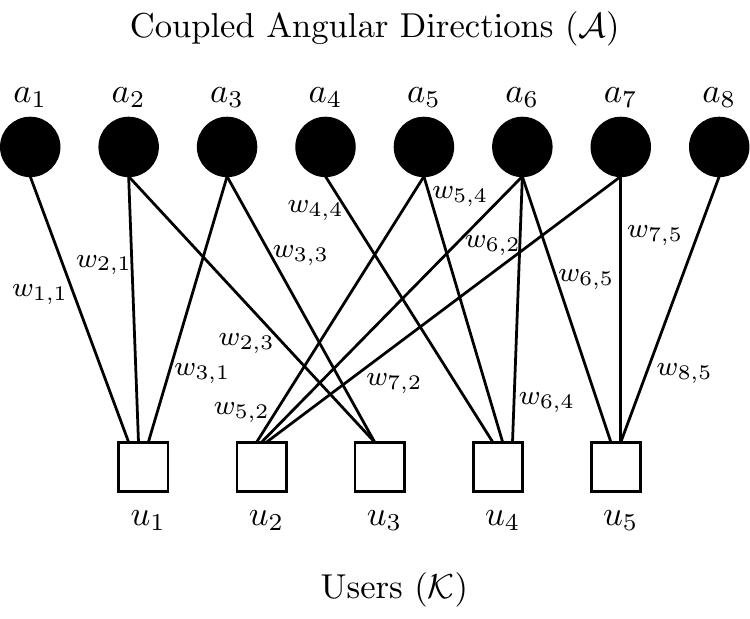}
		\caption{}
		\label{fig:bi_graph}
	\end{subfigure}%
	\begin{subfigure}[b]{.35\textwidth}
		\centering
		\includegraphics[width=.55\linewidth]{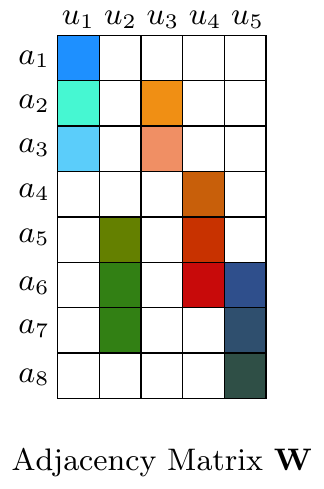}
		\caption{}
		\label{fig:adj_mat}
	\end{subfigure}
	\caption{\small (a) An example of a bipartite graph $\Lc$. (b) The corresponding weighted adjacency matrix $\Wm$.}
	\label{fig:test}
\end{figure}
\begin{problem}\label{problem:opt_prob_1}
	Let $\Tdl$ denote the available DL pilot dimension and let $\Mc (\Ac',\Kc')$ denote a matching of the subgraph $\Lc' (\Ac',\Kc',\Ec')$ 
	of the bipartite graph $\Lc (\Ac,\Kc,\Ec)$. Find the solution of the following optimization problem:	
	\begin{subequations}\label{eq:my_opt_1}
		\begin{align}
			& \underset{\Ac'\subseteq \Ac, \Kc' \subseteq \Kc}{\text{maximize}} && \left\vert \Mc\left(\Ac',\Kc'\right)\right\vert    \label{eq:my_opt_1-one}  \\
			& \text{subject to} &&  \text{deg}_{\Lc'} (k) \le \Tdl~ \forall k\in \Kc',  \label{eq:my_opt_1-two} \\
			& ~ &&  \hspace{-4mm} \sum_{a\in \Nc_{\Lc'} (k) } w_{a,k}  \ge \Pthresh, ~ \forall k\in \Kc'. \label{eq:my_opt_1-three}
		\end{align}
	\end{subequations} 
\hfill $\lozenge$
\end{problem}
The following theorem shows that Problem \ref{problem:opt_prob_1} can be solved in a tractable way.
\begin{theorem}\label{thm:MILP_formulation}
	The optimization problem in \eqref{eq:my_opt_1} is equivalent to the mixed integer linear program (MILP) below: 
	 
	 \begin{subequations}\label{opt:P_MILP_Thm}
	 	\begin{align}
	 	\Pc_{\text{MILP}}:~~	\underset{x_m,y_k,z_{m,k} }{\text{maximize}}  & ~~\sum_{m\in \Ac} \sum_{k \in \Kc} z_{m,k}   \label{eq:obj_1_thm} \\
	 	\text{subject to}  &  ~~~~z_{m,k} \le [\Am]_{m,k} ~~\forall m\in\Ac,k\in \Kc, \label{eq:first_thm}\\
	 	~~ & ~~~~ \sum_{k\in \Kc}z_{m,k} \le x_m ~~ \forall m\in \Ac,\\ 
	 	~~ & ~~~~ \sum_{m\in \Ac}z_{m,k} \le y_k ~~ \forall k\in \Kc,  \\
	 	~~ & ~~~~  \sum_{m\in \Ac} [\Am]_{m,k} x_m \le \Tdl y_k + M (1-y_k) ~~\forall k\in \Kc,\\
	 	~~ & ~~~~  \Pthresh \, y_k \le \sum_{m \in \Ac} [\Wm]_{m,k} x_m ~~\forall k \in \Kc,  \\
	 	~~ & ~~~~  x_m \le \sum_{k\in \Kc} [\Am]_{m,k} y_k ~~ \forall m\in \Ac, \\
	 	~~ & ~~~~  x_m, y_k \in \{0,1\} ~~\forall a\in \Ac,k\in \Kc,\\
	 	~~ & ~~~~  z_{m,k} \in [0,1] ~~\forall m\in \Ac,k\in \Kc \label{eq:last_thm},
	 	\end{align}
	 \end{subequations}
	 where $\Wm$ is the $|\Ac| \times |\Kc|$ weighted adjacency matrix in which $[\Wm]_{m,k} = w_{m,k}$ (see the example in Fig. \ref{fig:bi_graph} and Fig. \ref{fig:adj_mat}). The solution sub-graph is given by the set of nodes $\Ac' = \{m:x_m^\ast = 1\}$ and $\Kc' = \{k:y_k^\ast=1\}$, with $\{x_m^\ast\}_{m=1}^M$ and $\{y_k^\ast\}_{k=1}^K$ being a solution of \eqref{opt:P_MILP_Thm}. \hfill $\square$
\end{theorem}
\begin{proof}
	See Appendix \ref{app:proof_milp_formulation}.
\end{proof}
The introduced MILP can be efficiently solved using an off-the-shelf optimization toolbox. The solution to this optimization, however, is not necessarily unique, i.e. there may exist several sub-graphs with the same (maximum) matching size. In order to limit the solution set we introduce a regularization term to the objective of \eqref{opt:P_MILP_Thm} to favor solutions containing more ``active" beams. The regularized form of \eqref{opt:P_MILP_Thm} is given as
\begin{equation}\label{opt:P_MILP_Reg}
\begin{aligned}
\Pc_{\text{MILP}}:~~	\underset{x_m,y_k,z_{m,k} }{\text{maximize}}  & ~~\sum_{m\in \Ac} \sum_{k \in \Kc} z_{m,k}   + \epsilon \sum_{m \in \Ac} x_m \\
\text{subject to}  &  ~~\{x_m,y_k,z_{m,k}\}_{m\in \Ac,k\in \Kc} \in \Sc_{\text{feasible}}, \\
\end{aligned}
\end{equation}  
where the feasibility set $\Sc_{\text{feasible}}$ encodes the constraints \eqref{eq:obj_1_thm}-\eqref{eq:last_thm}.
Here the regularization factor $\epsilon$ is chosen to be a small positive value such that it does not effect the matching size of the solution sub-graph. In fact choosing $\epsilon<\frac{1}{M}$ ensures this, since then $ \epsilon \sum_{m \in \Ac} x_m < 1$ and a solution to \eqref{opt:P_MILP_Reg} must have the same matching size as a solution to \eqref{opt:P_MILP_Thm}, otherwise the objective of \eqref{opt:P_MILP_Reg} can be improved by choosing a solution with a larger matching size. 

\subsection{Channel estimation and multiuser precoding}  \label{sec:precoding}  

For a given set of user DL covariance matrices, let $\left\{x_m^\ast \right\}_{m=1}^M$ and $\left\{y_k^\ast \right\}_{k=1}^K$ denote the MILP solution and
denote by $\Bc=\{m:x_m^\ast=1\} = \{m_1,m_2,\ldots, m_{M'}\}$ the set of selected beam directions of cardinality $|\Bc| = M'$ and 
by $\Kc = \{k: y^\ast_k = 1\}$ the set of selected users of cardinality $|\Kc| = K'$. 
The resulting sparsifying precoding matrix $\Bm$ in (\ref{eq:cs_eq_1}) is simply obtained 
as $\Bm = \Fm_{\Bc}^\herm$, where $\Fm_{\Bc} = [\fv_{m_1}, \ldots, \fv_{m_{M'}}]$ 
and $\fv_{m}$ denotes the $m$-th column of the $M \times M$ unitary DFT matrix $\Fm$.  
Given a DFT column $\fv_m$,  we have 
\[ \Bm \fv_m  = \left \{ \begin{array}{ll}
\zerov & \mbox{if} \;\; m \notin \Bc \\
\uv_i & \mbox{if} \;\; m = m_i \in \Bc 
\end{array} \right . \]
where $\uv_i$ denotes a $M' \times 1$ vector with all zero components but a single ``1'' in the $i$-th position.
Using the above property and (\ref{approximate-KL}), the effective DL channel vectors take on the form
\begin{equation} \label{effch}
\heff^{(k)}  = \Bm \sum_{m \in \Sc_k} g_m^{(k)} \sqrt{[\lambdav^{(k)}]_m} \fv_m =  \sum_{i : m_i \in \Bc \cap \Sc_k}  \sqrt{[\lambdav^{(k)}]_{m_i}} g_{m_i}^{(k)} \uv_i.
\end{equation}
In words, the effective channel of user $k$ is a vector with non-identically zero elements only at the positions corresponding to 
the intersection of the beam directions in $\Sc_k$, along which the physical channel of user $k$ carries positive energy, 
and in $\Bc$, selected by the sparsifying precoder. 
The non-identically zero elements are independent Gaussian coefficients $\sim \Cc\Nc(0, [\lambdav^{(k)}]_{m_i})$.  
Notice also that, by construction, the number of non-identically zero coefficients are $|\Bc \cap \Sc_k| \leq \Tdl$ and their 
positions (encoded in the vectors $\uv_i$ in (\ref{effch})), plus an estimate of their variances $[\lambdav^{(k)}]_{m_i}$ are known to the BS. 
Hence, the effective channel vectors can be estimated from the $\Tdl$-dimensional DL pilot observation (\ref{eq:cs_eq_1}) with an estimation MSE
that vanishes as $1/\SNR$. The pilot observation in the form (\ref{eq:cs_eq_1}) is obtained at the 
user $k$ receiver. In this work, we assume that each user sends its pilot observations using $\Tdl$ channel uses in the UL, using analog unquantized feedback, as analyzed 
for example in \cite{caire2010multiuser,kobayashi2011training}. 
At the BS receiver, after estimating the UL channel from the UL pilots, 
the BS can apply linear MMSE estimation and recovers the channel state feedback which takes on the same form of (\ref{eq:cs_eq_1}) with some additional noise
due to the noisy UL transmission.\footnote{As an alternative, one can consider quantized feedback using $\Tdl$ channel uses 
in the UL (see \cite{caire2010multiuser,kobayashi2011training} and references therein). Digital quantized feedback yields generally a better end-to-end estimation MSE in the absence of
feedback errors. However, the effect of decoding errors on the channel state feedback is difficult to characterize in a simple manner since
it depends on the specific joint source-channel coding scheme employed. Hence, in this work we restrict to the  simple analog feedback.} 

With the above precoding, we have $\Bm \Bm^\herm = \Id_{M'}$. Also, 
we can choose the DL pilot matrix $\Psim$ to be proportional to a random unitary matrix of dimension $\Tdl \times M'$, such that 
$\Psim \Psim^\herm = \Pdl  \Id_{\Tdl}$. In this way, the DL pilot phase power constraint (\ref{power-training-phase}) is automatically satisfied. 
The estimation of $\heff^{(k)}$ from the DL pilot observation (\ref{eq:cs_eq_1}) (with suitably increased AWGN variance due to the noisy UL feedback) is completely 
straightforward and shall not be treated here in details. 

For the sake of completeness, we conclude this section with the DL precoded data phase and the corresponding
sum rate performance metric that we shall use in Section \ref{sec:results} for numerical analysis and comparison with other schemes. 
Let $\widehat{\Hm}_{{\text{eff}}}= [ \widehat{{\hv}}_{{\text{eff}}}^{(1)}, \ldots, \widehat{{\hv}}_{{\text{eff}}}^{(K')}]$ be the matrix 
of the estimated effective DL channels for the selected users. We consider the ZF beamforming matrix $\Vm$ given by the column-normalized 
version of the Moore-Penrose pseudoinverse of the estimated channel matrix, i.e., $\Vm = 
\left(\widehat{\Hm}_{{\text{eff}}} \right)^\dagger \Jm^{1/2}$, where $\left(\widehat{\Hm}_{{\text{eff}}} \right)^\dagger = 
\widehat{\Hm}_{{\text{eff}}}  \left ( \widehat{\Hm}_{{\text{eff}}}^\herm \widehat{\Hm}_{{\text{eff}}} \right )^{-1}$ and $\Jm$ is a diagonal matrix that makes the columns of 
$\Vm$ to have unit norm. A channel use of the DL precoded data transmission phase at the $k$-th user receiver takes on the form
\begin{equation} \label{receivedk}
y^{(k)} = \left ( \hv^{(k)} \right )^\herm \Bm^\herm \Vm \Pm^{1/2} \dv + n^{(k)}, 
\end{equation} 
where $\dv \in \CC^{K' \times 1}$ is a vector of unit-energy user data symbols and $\Pm$ is a diagonal matrix defining the 
power allocation to the DL data streams. The transmit power constraint is given by 
\[ \trace( \Bm^\herm \Vm \Pm \Vm^\herm \Bm ) = \trace ( \Vm^\herm \Vm \Pm ) = \trace (\Pm) = \Pdl, \]
where we used $\Bm \Bm^\herm = \Id_{M'}$ and the fact that $\Vm^\herm \Vm$ has unit diagonal elements by construction. 
In particular, in the results of Section \ref{sec:results} we use the simple uniform power allocation $P_k = \Pdl/K'$ to each $k$-th user data stream. In the case of perfect ZF beamforming, i.e., for  $\widehat{\Hm}_{{\text{eff}}} = \Hm_{{\text{eff}}}$, we have that (\ref{receivedk}) reduces to $y^{(k)} = \sqrt{J_k P_k} d_k  +  n^{(k)}$,
where $J_k$ is the $k$-th diagonal element of the norm normalizing matrix $\Jm$, 
$P_k$ is the $k$-th diagonal element of the power allocation matrix $\Pm$, and $d_k$ is the $k$-th user data symbol. 
Since in general $\widehat{\Hm}_{{\text{eff}}} \neq \Hm_{{\text{eff}}}$, due to non-zero estimation error, the received symbol at user $k$ receiver is given by $ y^{(k)} = b_{k,k} d_k + \sum_{k' \neq k} b_{k,k'} d_{k'}   +  n^{(k)},$
where the coefficients $(b_{k,1}, \ldots, b_{k,K'})$ are given by the elements of the $1 \times K'$ row vector
$\left ( \hv^{(k)} \right )^\herm \Bm^\herm \Vm \Pm^{1/2}$ in (\ref{receivedk}). Of course, in the presence of an accurate channel 
estimation we expect that $b_{k,k} \approx \sqrt{J_k P_k}$ and $b_{k,k'} \approx 0$ for $k' \neq k$. For simplicity, in this paper we compare the performance of the proposed scheme
with that of the state-of-the-art CS-based scheme in terms of ergodic sum rate, assuming that all coefficients 
$(b_{k,1}, \ldots, b_{k,K'})$ are known to the corresponding receiver $k$. Including the DL training overhead, this yields the rate expression (see \cite{DBLP:journals/corr/Caire17}) 
\begin{equation}\label{eq:rate_ub}
R_{\rm sum}  = \left (1 - \frac{\Tdl}{T} \right ) \sum_{k \in \Kc} \EE \left [ \log \left ( 1 + \frac{\left|b_{k,k}\right|^2}{1 + \sum_{k' \neq k} \left|b_{k,k'} \right|^2} \right ) \right ].
\end{equation}

\section{Simulation Results}  \label{sec:results}

In this section we compare the performance of the proposed 
approach for FDD massive MIMO 
to two of the most recent CS-based methods proposed in \cite{rao2014distributed} and \cite{ding2016dictionary} 
in terms of channel estimation error and sum-rate. 
In \cite{rao2014distributed}, the authors proposed a method based on common probing of the DL channel with random Gaussian pilots. 
The DL pilot measurements $\yv^{(k)}$ at users $k=1,\ldots,K$ are fed back and
collected by the BS, which recovers the channel vectors using a joint orthogonal matching pursuit (J-OMP) technique able to exploit the possible
common sparsity between the user channels (see channel model in Section \ref{sec:sys_setup}). 
 
	
	In \cite{ding2016dictionary}, a method based on dictionary learning for sparse channel estimation was proposed. 
	In this scheme, the BS jointly \textit{learns} sparsifying dictionaries for the UL and DL channels by collecting channel measurements 
	at different cell locations (e.g., via an off-line learning phase). 
	The actual user channel estimation is posed as a norm-minimization convex program using the trained dictionaries 
	and with the constraint that UL and DL channels share the same support over their corresponding dictionaries. 
	Following the terminology used in \cite{ding2016dictionary}, we refer to this method as JDLCM.

For this comparison, we considered $M = 128$ antennas at the BS, $K=13$ users, and resource blocks of size $T = 128$ symbols. 
For our proposed method, the BS computes the users' sample UL covariance matrices by taking $N_{\rm ul}=1000$ UL pilot observations and then applies 
the scheme explained in Section \ref{sec:DL_cov_est}.  
Given the obtained DL channel covariance matrix estimates,  we first perform the circulant 
approximation and extract the vector of approximate eigenvalues as in  (\ref{eq:circ_approx_2}).
Then, we compute the sparsifying precoder $\Bm$ 
via the MILP solution as given in Section \ref{sec:sparseB}.  
In the results presented here,  we set the parameter $\Pthresh$ in the MILP to a small value 
in order to favor a high rank of the resulting 
effective channel matrix over the beamforming gain.\footnote{This approach is appropriate in the medium to high-SNR regime. For low SNR, it is often convenient to
	increase $\Pthresh$ in order to serve less users with a larger beamforming energy transfer per user.} 
	After probing the effective channel of the selected users along these active beam directions 
via a random unitary pilot matrix $\Psim$, we calculate their MMSE estimate using the estimated DL covariance matrices.
Eventually, for all the three methods, we compute the ZF beamforming matrix based on the obtained channel estimates. 
In addition, instead of considering all selected users, in both cases we apply the Greedy ZF user selection approach of \cite{dimic2005downlink}, 
that yields a significant benefit when the number of users is close to the rank of the effective channel matrix.  
As said before, the DL SNR is given by $\SNR = \Pdl/\Nvar$ and during the simulations we consider ideal noiseless feedback 
for simplicity, i.e., we assume that the BS receives the measurements in \eqref{eq:cs_eq_1} without extra feedback noise 
to the system.\footnote{Notice that by 
	introducing noisy feedback the relative gain w.r.t. J-OMP is even larger, since CS schemes are known to be more noise-sensitive than plain MMSE estimation using estimated DL covariance matrices.} 
	The sparsity order of each channel vector is given as an input to the J-OMP method,
but not to the other two methods. This represents a genie-aided advantage for J-OMP,  that we introduce here for simplicity.   

\begin{figure}[t]
	\centering
	\begin{subfigure}[b]{.4\textwidth}
		\centering
		\includegraphics[width=1.0\linewidth]{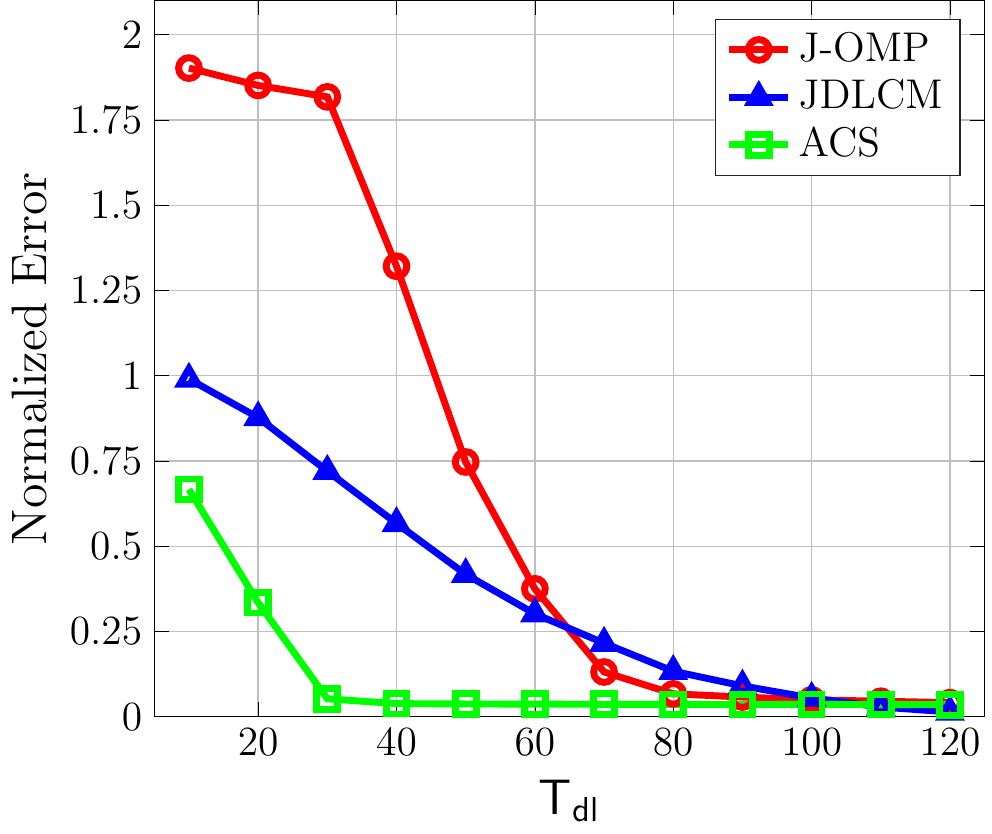}
		\caption{}
		\label{fig:err_plot_blocky}
	\end{subfigure}%
	\begin{subfigure}[b]{.4\textwidth}
		\centering
		\includegraphics[width=1\linewidth]{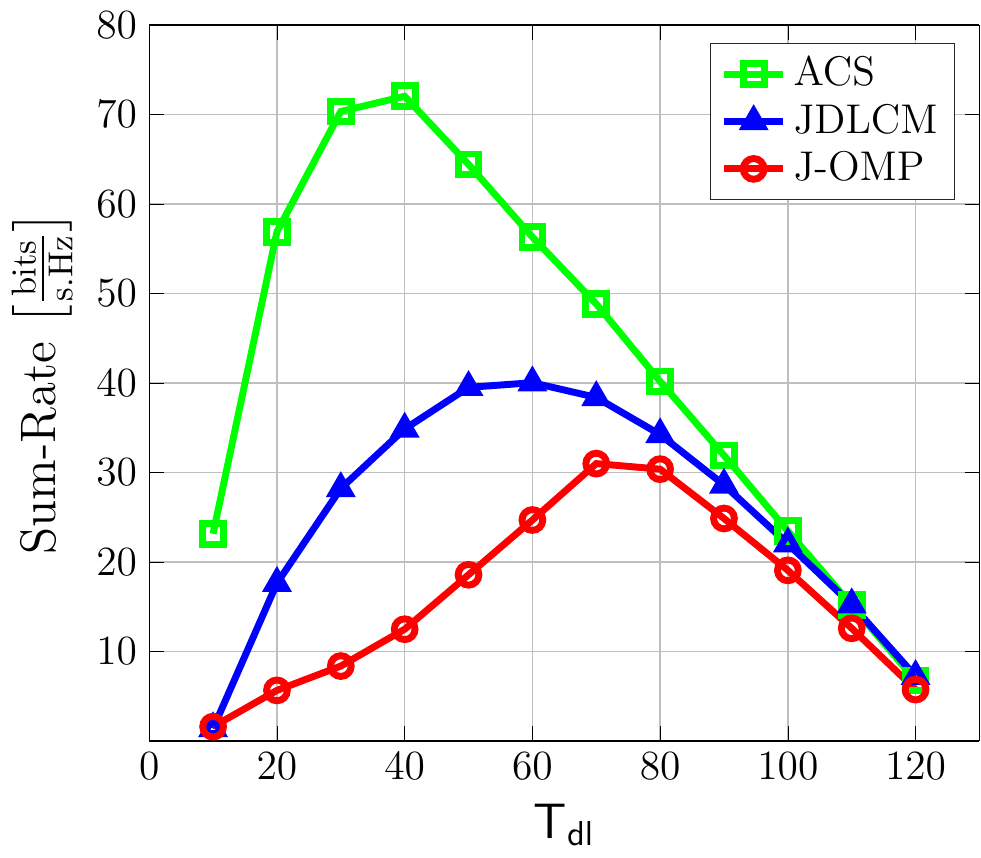}
		\caption{}
		\label{fig:sum_rate_blocky}
	\end{subfigure}
	\caption{\small (a) Normalized channel estimation error, and (b) achievable sum-rate as a function of DL pilot dimension with $\SNR = 20$ dB, $M=128$ and $K=13$.} 
	\label{fig:blocky}
\end{figure}
As the simulation geometry, we consider three MPC clusters with random locations within the angular range (parametrized by $\xi$ rather than $\theta$) $ [-1,1)$. We denote by $\Xi$ the $i$-th interval and set each interval size to be $|\Xi_i| = 0.2,~i=1,2,3$. 
	The ASF for each user is obtained by selecting at random two out of three such clusters, such that the overlap of the 
	angular components among users is large. 
	The ASF is non-zero over the angular intervals corresponding to the chosen MPCs and zero elsewhere, i.e., $\gamma_k (d\xi) = \beta {\bf 1}_{\Xi_{i_1} \cup \Xi_{i_2}},$ where $\beta = 1/\int_{-1}^{1} \gamma_k(d\xi)$ and $i_1,i_2\in \{1,2,3\}$. 
	The described arrangement results in each generated channel vector being roughly $s_k = 0.2\times M\approx 26$-sparse. To measure channel estimation error we use the normalized Euclidean distance as follows. Let $\Hm \in \bC^{M\times K'}$ define the matrix whose columns correspond to the channel vectors of the $K'$ served users and let $\widehat{\Hm}$ denote the estimation of $\Hm$. Then the normalized error is defined as
	\[e = \bE\left[\frac{\Vert \Hm - \widehat{\Hm}\Vert^2}{\Vert \Hm \Vert^2}\right]. \]
	
\subsection{Comparisons} \label{sec:diff_scat}
	
	Fig. \ref{fig:err_plot_blocky} shows the normalized channel estimation error for the J-OMP, JDLCM and our proposed Active Channel Sparsification (ACS) method as a function of the DL pilot dimension $\Tdl$ with $\SNR = 20$ dB. Our ACS method outperforms the other two by a large margin, 
	especially for low DL pilot dimensions. When the pilot dimension is below channel sparsity order, CS-based methods perform very poorly, since the number of channel measurements is less than the inherent channel dimension. 
	Fig. \ref{fig:sum_rate_blocky} compares the achievable sum-rate for the three methods. Again our ACS method shows a much better performance 
	compared to J-OMP and JDLCM. This figure also shows that there is an optimal DL pilot dimension that maximizes the sum-rate. 
	This optimal value is $\Tdl\approx 40$ for our proposed method, $\Tdl \approx 60$ for JDLCM and $\Tdl\approx 70$ for the J-OMP method.
\subsection{The effect of channel sparsity order}
	
\begin{figure}[t]
	\centering
	\includegraphics[width=0.45\linewidth]{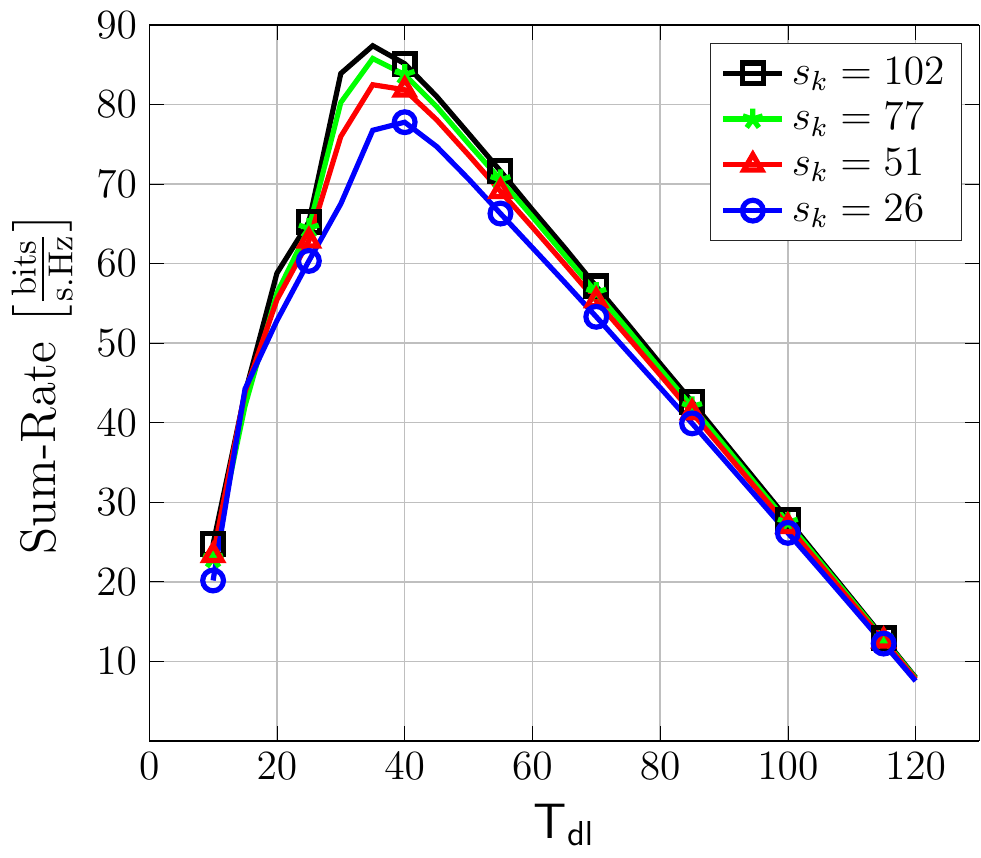}
	\caption{\small Sum-rate vs $\Tdl$ for various channel sparsity orders. Here $\SNR = 20$ dB, $M=128$ and $K=13$.}
	\label{fig:rate_sparsity_compare}
\end{figure}

	Depending on the geometry and user location, channels may show different levels of sparsity in the angular domain. In contrast to CS-based methods, our proposed method is highly flexible with regards to various channel sparsity orders, thanks to the active sparsification method. In this section, we investigate how sparsity order effects channel estimation error as well as sum-rate within the framework of our proposed method. We use the same setup as in section \ref{sec:diff_scat}, i.e. user ASFs consist of two clusters chosen at random among the three. But now we vary the size of the angular interval each of the clusters occupies ($|\Xi_i|=0.2,0.4,0.6,0.8$) and see how it effects the error and sum-rate metrics. The sparsification, channel probing and transmission are performed as described before. Since each ASF consists of two clusters and $M=128$ channel sparsity order (roughly) takes on the values $s_k=26,51,77,102$ for all users $k\in [K']$. For each value of the pilot dimension we perform a Monte Carlo simulation to empirically calculate the sum-rate. Fig. \ref{fig:rate_sparsity_compare} illustrates the results. Notice that in these results we fix the channel coefficient power along each scattering component, such as richer (less sparse) channels convey more signal energy. This corresponds to the physical fact that the more scattered signal energy is collected at the receiving antennas the higher the received signal energy is.  As we can see in Fig. \ref{fig:rate_sparsity_compare}, for a fixed $\Tdl$, when the number of non-zero channel coefficients increases (i.e., the channel is less sparse), we generally have a larger sum-rate. The main reason is that, with less sparse channels, the beamforming gain is larger due to the fact that more scattering components contribute to the channel. Therefore, we can generally say that with our method, for a fixed pilot dimension, 
less sparse channels are {\em better}. Of course, this is not the case for CS-based techniques, or techniques based on 
the ``sparsity assumption'' of a small number of discrete angular components, which tend to collapse and yield very bad results when such sparsity assumptions are not satisfied.

\begin{figure}[t]
	\centering
	\includegraphics[width=0.45\linewidth]{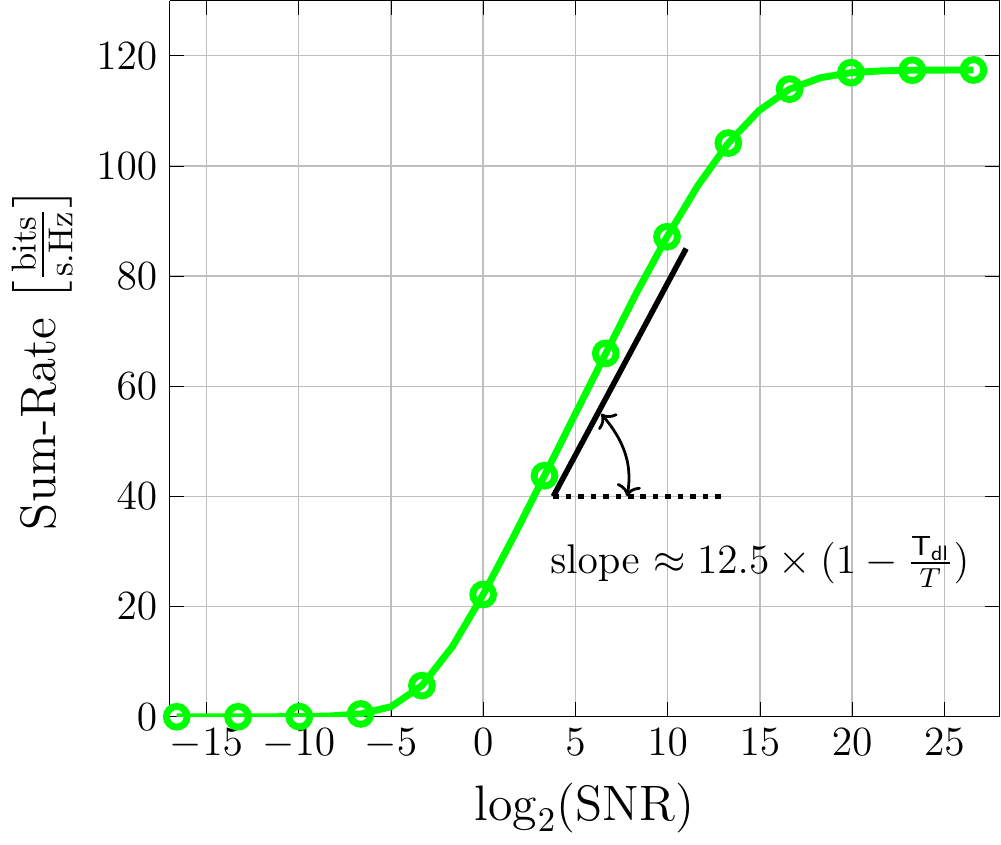}
	\caption{\small Sum-rate as a function of $\log_2\text{(SNR)}$ with $M=128$ and $K=13$.}
	\label{fig:mplex_gain}
\end{figure}

\subsection{Relevance of the pre-log factor}

An interesting final observation is to examine the system sum-rate vs. SNR with 
our proposed method, and in particular show that there is indeed a regime of intermediate SNR for which the slope of the sum-rate curve yields
quite faithfully the number of spatially multiplexed data streams. 
We performed a simulation with $M=128$ antennas and $K=13$ users and a pilot dimension of $\Tdl = 60$. 
The pre-log factor determines the slope of the sum-rate vs $\log_2(\text{SNR})$ curve, in an intermediate regime where the sum-rate is not saturated, and yet the spectral efficiency is large.\footnote{This saturation is due to the non-vanishing covariance estimation error and happens at around $\text{SNR}=60$ dB.} 
As illustrated in Fig. \ref{fig:mplex_gain}, this slope is equal to $12.5 \times (1-\frac{\Tdl}{T})$. Notice that the Greedy ZF scheme decides to serve
a number of users that may be less than $K$ in an opportunistic fashion, such that the expected number of 
served users (DL data streams) in this SNR regime is indeed slightly less than the maximum possible $K = 13$. Hence, the agreement between
the sum-rate slope in this regime and the number of served DL data streams is exactly what can be expected, thus showing the relevance of 
maximizing the rank of the effective matrix in the proposed optimization of the sparsifying precoder.
\section{Conclusion}

We presented a novel approach for FDD massive MIMO systems. 
Our approach exploits the reciprocity of the angular scattering function to estimate the 
covariance matrix of the users' DL channels from the UL pilots 
sent by the users to the BS.  The estimated DL covariance matrices of all users can be approximately expressed in terms of a  
common system of covariance eigenvectors (beam-space representation).  For the ULA setting, such eigenvectors 
are the columns of a DFT matrix, and this representation incurs a vanishing error for large number of BS antennas $M$.
This beam-space information allows the BS to smartly select a set of beams and users such that communication over the resulting 
effective channels is efficient even with a limited DL pilot dimension.  
This beam-user selection procedure is referred to here as 
{\em active channel sparsification} and is achieved via a newly formulated mixed integer linear program (MILP). 
Our simulation results show that the proposed method performs well even in cases where the available DL pilot dimension 
is far less than the inherent dimension of the channel vectors. This represents a fundamental improvement with respect to the state-of-the-art CS-based 
method (in particular, exploiting common sparsity or learned sparsifying dictionaries), 
for which the DL pilot dimension should always be larger than the inherent channel sparsity in the angle domain. 
We conclude by mentioning that in this paper we focused on purpose on a simple single-cell scenario. 
When multiple cells are considered, inter-cell interference should be taken into account. 
However, unlike TDD systems where UL and DL across different cells are synchronous, 
and the limited pilot dimension yields pilot contamination (see \cite{Marzetta-TWC10,larsson2014massive,marzetta2016fundamentals}), in FDD systems
there is no need for tight inter-cell synchronization and the inter-cell incoherent interference simply results 
in a higher level of the background noise,  but can be taken into account in a completely straightforward manner 
(as always traditionally done in the analysis of cellular systems) since no coherently beamformed interference due to pilot contamination 
appears in FDD systems. 

Future work along the lines presented in this paper 
may consist of generalizing the active channel sparsification method to a broader category of array geometries. 
While such generalization is straightforward for UPAs, leveraging the block-Toeplitz covariance structure, for other geometries one must find
efficient methods for UL-DL covariance transformation and efficient ``beam-space representation'' for the design of the sparsifying precoder.
\section{Appendices}
\subsection{Proof of Lemma \ref{lem:stable_rec}}\label{App:stable_rec_proof}

The proof follows by using the representation 
$\hv^{(k)} =  \sum_{m \in \Sc_k} g_{m}^{(k)} \sqrt{[\lambdav^{(k)}]_m} \fv_m$
(see (\ref{approximate-KL})), which holds exactly by assumption. 
	Estimating $\hv^{(k)}$ is equivalent to estimating the vector of KL Gaussian i.i.d. 
	coefficients $\gv^{(k)} = (g_{m}^{(k)} : m \in \Sc_k) \in \CC^{s_k \times 1}$. 
	Define the $M \times s_k$ DFT submatrix $\Fm_{\Sc_k} = (\fv_m : m \in \Sc_k)$, and the corresponding
	diagonal $s_k \times s_k$ matrix of the non-zero eigenvalues $\Lambdam^{(k)}_{\Sc_k}$.
After some simple standard algebra, the MMSE estimation error covariance of $\gv^{(k)}$ from ${\bfy}^{(k)}$ 
in (\ref{eq:cs_eq_1}) with $\Bm = \Id_M$ can be written in the form 
	\begin{eqnarray} 
	\widetilde{\Rm}_e &= &  \Id_{s_k} - 
	\left ( \Lambdam^{(k)}_{\Sc_k} \right )^{1/2}  \Fm_{\Sc_k}^\herm \Psim^\herm 
	\left (    \Psim  \Fm_{\Sc_k} \Lambdam^{(k)}_{\Sc_k}  \Fm_{\Sc_k}^\herm \Psim^\herm + \Nvar \Id_{\Tdl} \right )^{-1} 
	\Psim \Fm_{\Sc_k} \left ( \Lambdam^{(k)}_{\Sc_k} \right )^{1/2}.    \label{error-cov}
	\end{eqnarray}
Using the fact that  $\Rm_e =  \Fm_{\Sc_k} ( \Lambda_{\Sc_k}^{(k)})^{1/2} \widetilde{\Rm}_e ( \Lambda_{\Sc_k}^{(k)})^{1/2} \Fm_{\Sc_k}^\herm$, 
such that $\trace(\Rm_e) = \trace(\Lambdam_{\Sc_k}  \widetilde{\Rm}_e)$, we have that 
$\trace(\Rm_e)$ and $\trace(\widetilde{\Rm}_e)$ have the same vanishing order with respect to $\Nvar$. In particular, 
it is sufficient to consider the behavior of $\trace(\widetilde{\Rm}_e)$ as a function of $\Nvar$. Now, 
using the Sherman-Morrison-Woodbury  matrix inversion lemma \cite{horn1990matrix}, 
after some algebra omitted for the sake of brevity we arrive at
\begin{equation} 
\trace(\widetilde{\Rm}_e) = s_k - \sum_{i=1}^{s_k} \frac{\mu_i}{\Nvar+ \mu_i},   \label{limitN0}
\end{equation}
where $\mu_i$ is the $i$-th eigenvalue of the $s_k \times s_k$ matrix 
$\Am = ( \Lambda_{\Sc_k}^{(k)})^{1/2} \Fm^\herm_{\Sc_k} \Psim^\herm \Psim \Fm_{\Sc_k} ( \Lambda_{\Sc_k}^{(k)})^{1/2}$.  
Next, notice that 
\begin{equation} \label{rank-upper}
{\rm rank}(\Am) = {\rm rank} (\Fm^\herm_{\Sc_k} \Psim^\herm \Psim \Fm_{\Sc_k} ) = 
{\rm rank} (\Fm_{\Sc_k}  \Fm^\herm_{\Sc_k} \Psim^\herm) \leq \min\{ s_k , \Tdl\}. 
\end{equation}
In fact, $\Lambda_{\Sc_k}^{(k)}$ is diagonal with strictly positive diagonal elements, 
such that left and right multiplication by $( \Lambda_{\Sc_k}^{(k)})^{1/2}$ yields rank-preserving 
row and column scalings,  the matrix $\Fm_{\Sc_k}  \Fm^\herm_{\Sc_k}$ is the orthogonal projector onto the 
$s_k$-dimensional column-space of $\Fm_{\Sc_k}$ and has rank $s_k$, while the matrix  
$\Psim^\herm \in \CC^{M \times \Tdl}$ has the same rank of $\Psim^\herm \Psim$, that is at most $\Tdl$. 

For $\Tdl \geq s_k$ the existence of matrices $\Psim$ such that the rank 
upper bound (\ref{rank-upper}) holds with equality (i.e., for which rank$(A) = s_k$ for any support set $\Sc_k$ of size $s_k$) 
is shown as follows. Generate a random $\Psim$ with i.i.d. elements $\sim \Cc\Nc(0,1)$. 
Then,  the columns of $\Fm^\herm_{\Sc_k} \Psim^\herm$ form a collection of $\Tdl \geq s_k$ mutually
independent $s_k$-dimensional Gaussian vectors with i.i.d. $\sim \Cc\Nc(0,1)$ components. The event that 
these vectors span a space of dimension less than $s_k$ is a null event (zero probability). 
Hence, such randomly generated matrix satisfies the rank equality in (\ref{rank-upper}) with probability 1.
As a consequence,  for $\Tdl \geq s_k$ we have that $\mu_i > 0$ for all $i \in [s_k]$ and (\ref{limitN0}) vanishes as $O(\Nvar)$ as $\Nvar \downarrow 0$.
In contrast, if $\Tdl < s_k$, by (\ref{rank-upper}) for any matrix $\Psim$ at most $\Tdl$ eigenvalues $\mu_i$ in (\ref{limitN0}) are non-zero 
and $\lim_{\Nvar\downarrow 0} s_k - \sum_{i=1}^{s_k} \frac{\mu_i}{\Nvar + \mu_i} \geq s_k - \Tdl > 0 $. \hfill \QED

\subsection{Proof of Lemma \ref{lem:matching}}\label{app:proof_lemma_matching}

		The determinant of $\Qm$ is given by the expansion $\text{det}(\Qm)=\sum_{\iota\in \boldsymbol{\pi}_r} \text{sgn}(\iota) \prod_{i} [\Qm]_{i,\iota (i)}$, where $\iota$ is a permutation of the set $\{1,2,\ldots,r\}$, where $\boldsymbol{\pi}_r$ is the set of all such permutations and where $\text{sgn}(\iota)$ is either 1 or -1. The product $\prod_{i} [\Qm]_{i,\iota (i)}$ is non-zero only for the perfect matchings in the bipartite graph. Hence, if the bipartite graph contains a perfect matching, then $\text{det}(\Qm)\neq 0$ with probability 1 (and $\text{rank}(\Qm)=r$), since the non-identically zero entries of $\Wm$ are drawn from a continuous distribution. If it does not contain a perfect matching, then $\text{det}(\Qm)=0$ and therefore $\text{rank}(\Qm)<r$. \hfill \QED

\subsection{Proof of Theorem \ref{thm:MILP_formulation}\label{app:proof_milp_formulation}}
First, without loss of generality let assume that 
$\Lc$ contains no isolated nodes (since these would be discarded anyway).  
As before the $|\Ac|\times |\Kc|$ weighted adjacency matrix is denoted by $\Wm$ where $[\Wm]_{m,k} = w_{m,k}$. 
An example of the bipartite graph $\Lc$ and its corresponding weighted adjacency matrix $\Wm$ is illustrated in Figs. \ref{fig:bi_graph} and \ref{fig:adj_mat}. Given the bipartite graph $\Lc(\Ac,\Kc,\Ec)$, we select the subgraph $\Lc'(\Ac',\Kc',\Ec')$, so that the constraint \eqref{eq:my_opt_1-two} is satisfied. 
We introduce the binary variables $\{x_m, m \in \Ac\}$ and $\{y_k, k \in \Kc\}$ to indicate if beam $m$ and user $k$ are selected, respectively.
As such, the constraint \eqref{eq:my_opt_1-two} is equivalent to the set of constraints:
\begin{subequations}
\begin{minipage}{.5\textwidth}
	\begin{align} 
x_m &\le \sum_{k\in \Kc} [\Am]_{m,k} y_k ~~ \forall m\in \Ac \label{eq-subgraph-selection-one}
	\end{align}
\end{minipage}%
\begin{minipage}{.5\textwidth}
	\begin{align} 
	y_k &\le \sum_{m \in \Ac} [\Am]_{m,k} x_m ~~\forall k \in \Kc  \label{eq-subgraph-selection-two} 
	\end{align}
\end{minipage}%
	\begin{align}
\sum_{m\in \Ac} [\Am]_{m,k} x_m &\le \Tdl y_k + M (1-y_k) ~~\forall k\in \Kc \label{eq-subgraph-selection-three}
\end{align}
\end{subequations}
In particular, \eqref{eq-subgraph-selection-one} ensures that if the beam $m$ is selected (i.e., $x_m=1$), 
there must be some $k \in \Kc$ such that $(m,k) \in \Ec$ is selected as well, whereas if beam $m$ is not selected, then this constraint is redundant. 
Similarly, in \eqref{eq-subgraph-selection-two} if user $k$ is selected (i.e., $y_k =1$), there must be some $m \in \Ac$ such that $(m,k) \in \Ec$ is selected as well. Furthermore, \eqref{eq-subgraph-selection-three} guarantees that if user $k$ is chosen (i.e., $y_k =1$), the number of chosen beams with $x_m=1$ is no more than $\Tdl$, and otherwise this constraint is redundant. Meanwhile, the constraint \eqref{eq:my_opt_1-three} is written as:
\begin{align} \label{eq:threshold}
\Pthresh \, y_k \le \sum_{m \in \Ac} [\Wm]_{m,k} x_m ~~\forall k \in \Kc
\end{align}
which ensures that if user $k$ is chosen (i.e., $y_k=1$) then the sum weights of the selected beams (i.e., $m \in \Nc_{\Lc'}(k)$ if $x_m=1$) is no less than $\Pthresh$, while if user $k$ is not chosen (i.e., $y_k=0$) then this constraint is not required and redundant. A closer look reveals that the constraint \eqref{eq:threshold} renders the one \eqref{eq-subgraph-selection-two} redundant, because when $y_k=1$ in \eqref{eq:threshold} there must exist at least one $m \in \Ac$ with $x_m=1$. Second, given the selected subgraph $\Lc'(\Ac',\Kc',\Ec')$, we find a matching  $\Mc (\Ac',\Kc')$ with maximum cardinality. 
To this end, we introduce another set of binary variables $\{z_{mk}, m \in \Ac, k \in \Kc\}$ to indicate if an edge $(a,k) \in \Ec$ is chosen to 
form the maximum matching  in $\Lc'(\Ac',\Kc',\Ec')$.
Following the canonical linear program formulation of the maximum cardinality matching for bipartite graphs, we translate the objective 
in \eqref{eq:my_opt_1} into the following optimization:
\begin{subequations} \label{eq:subgraph-matching}
	\begin{align}
	\underset{z_{m,k}\in \{0,1\} }{\text{maximize}}  & ~~\sum_{m \in \Ac'} \sum_{k \in \Kc'} [\Am]_{m,k} z_{m,k} \label{eq:sub_matching-one} \\
	\text{subject to}  &  ~~~~ \sum_{k\in \Kc'} [\Am]_{m,k} z_{m,k} \le 1 ~~\forall m\in \Ac', \label{eq:sub_matching-two}  \\
	~~ & ~~~~ \sum_{m\in \Ac'} [\Am]_{m,k} z_{m,k} \le 1 ~~ \forall k\in \Kc', \label{eq:sub_matching-three}
	\end{align}
\end{subequations}
Now, to transport the optimization problem on $\Lc'$ to the original setting on $\Lc$, we need to guarantee that $\Mc(\Ac',\Kc') \subseteq \Ec'$, i.e., $z_{mk}=1$ only if $m \in \Ac'$ ($x_m=1$), and $k \in \Kc'$ ($y_k=1$). This is obtained for a given configuration of the variables
$\{x_m\}$ and $\{y_k\}$ which define $\Lc'$, by adding constraints to \eqref{eq:subgraph-matching} and yields
\begin{subequations} \label{eq:matching}
	\begin{align}
	\underset{z_{m,k}\in \{0,1\} }{\text{maximize}}  & ~~\sum_{m \in \Ac} \sum_{k \in \Kc} [\Am]_{m,k} z_{m,k} \label{eq:matching-one} \\
	\text{subject to}  &  ~~~~ \sum_{k\in \Kc} [\Am]_{m,k} z_{m,k} \le 1 ~~\forall m\in \Ac, \label{eq:matching-two}  \\
	~~ & ~~~~ \sum_{m\in \Ac} [\Am]_{m,k} z_{m,k} \le 1 ~~ \forall k\in \Kc, \label{eq:matching-three}  \\
	~~ & ~~~~[\Am]_{m,k} z_{m,k} \le x_m ~~ \forall k \in \Kc, m \in \Ac,  \label{eq:matching-four} \\
	~~ & ~~~~[\Am]_{m,k} z_{m,k} \le y_k ~~ \forall k \in \Kc, m \in \Ac, \label{eq:matching-five}
	\end{align}
\end{subequations}
where \eqref{eq:matching-four}-\eqref{eq:matching-five} impose that the edge set $\{(m,k): z_{m,k}=1\}$ should be a subset of $\Ec'$. A further inspection on these constraints yields the following equivalent simplified form:
\begin{subequations} \label{eq:matching-simply}
	\begin{align} 
	\underset{z_{m,k}\in \{0,1\} }{\text{maximize}}  & ~~\sum_{m \in \Ac} \sum_{k \in \Kc} z_{m,k} \\
	\text{subject to}  &  ~~~~ z_{m,k} \le [\Am]_{m,k}, ~~ \forall m \in \Ac, k \in \Kc,   \label{eq:z-of-interest} \\
	~~&  ~~~~\sum_{k\in \Kc} z_{m,k} \le x_m, ~~\forall m\in \Ac, \label{eq:matching-two-four} \\ 
	~~&  ~~~~\sum_{m\in \Ac} z_{m,k} \le y_k, ~~ \forall k\in \Kc,  \label{eq:matching-three-five}
	\end{align} 
\end{subequations}
where the additional constraint \eqref{eq:z-of-interest} turns all the terms of the type $[\Am]_{m,k} z_{m,k}$ in \eqref{eq:matching} to $z_{m,k}$ in \eqref{eq:matching-simply}, the constraint \eqref{eq:matching-two-four} results from the combination of the constraints \eqref{eq:matching-two} and \eqref{eq:matching-four}, and \eqref{eq:matching-three-five} results from the combination of \eqref{eq:matching-three} with \eqref{eq:matching-five}. The formulation in \eqref{eq:matching-simply} can be seen as a modified maximum cardinality bipartite matching with selective vertices, 
in which the vertices with $x_m=1$ and $y_k=1$ are selected to participate in the maximum cardinality matching. The eventual mixed integer linear program is given as in \eqref{opt:P_MILP_Thm}. 
Notice that we have relaxed the binary constraint on $\{z_{m,k},\, m\in \Ac, k\in \Kc\}$ to the linear constraint \eqref{eq:last_thm} based on the following lemma. 
\begin{lemma}\label{thm:theorem_1}
	The problem $\Pc_{\rm MILP}$ as stated in \eqref{opt:P_MILP_Thm} always has binary-valued solutions for $\{z_{m,k},\, m\in \Ac, k\in \Kc\}$.  $\;\;\;\;$ \hfill $\square$
\end{lemma}
\begin{proof}
	It suffices to show that $z_{m,k}$ are binary, given that $x_m$ and $y_k$ are binary. First, if either $x_m,~ m\in\Ac$ or $y_k,~k\in\Kc$ are $0$, then $z_{a,k}=0$. So, we only need to focus on the case where $x_m=y_k=1,~m\in\Ac,k\in\Kc$. In that case, the constraints of $\Pc_{\rm MILP}$ with respect to $z_{m,k},~m\in\Ac,k\in\Kc$ form a convex polytope. This polytope is called the bipartite matching polytope, which is integral, i.e. all of its extreme points have integer (and in this case binary) values (see \cite[Corollary 18.1b. and Theorem 18.2.]{schrijver2003combinatorial}). Therefore, given $x_m,y_k\in \{0,1\},~\forall m\in\Ac,k\in\Kc$, $\Pc_{\rm MILP}$ reduces to a linear program with respect to the variables $z_{m,k}$ and the optimal solutions are the integral extreme points of the corresponding polyhedra and the proof is complete. 
\end{proof}

%

\balance

\bibliographystyle{IEEEtran}
{\footnotesize
\bibliography{references}
}

\end{document}

%% file: symbols.tex
\def\mindex#1{\index{#1}}

\def\sq{\hbox{\rlap{$\sqcap$}$\sqcup$}}
\def\qed{\ifmmode\sq\else{\unskip\nobreak\hfil
\penalty50\hskip1em\null\nobreak\hfil\sq
\parfillskip=0pt\finalhyphendemerits=0\endgraf}\fi\medskip}

\long\def\defbox#1{\framebox[.9\hsize][c]{\parbox{.85\hsize}{%
\parindent=0pt
\baselineskip=12pt plus .1pt      
\parskip=6pt plus 1.5pt minus 1pt 
 #1}}}

\long\def\beginbox#1\endbox{\subsection*{}%
\hbox{\hspace{.05\hsize}\defbox{\medskip#1\bigskip}}%
\subsection*{}}

\def\endbox{}

\newsavebox{\junk}
\savebox{\junk}[1.6mm]{\hbox{$|\!|\!|$}}

\def\bC{{\mathbb C}}

\def\bE{{\mathbb E}}

\def\bG{{\mathbb G}}

\def\bR{{\mathbb R}}

\def\bfB{{\bf B}}

\def\bfh{{\bf h}}

\def\bfn{{\bf n}}

\def\bfy{{\bf y}}

\def\sfH{{\sf H}}

\def\bfmath#1{{\mathchoice{\mbox{\boldmath$#1$}}%
{\mbox{\boldmath$#1$}}%
{\mbox{\boldmath$\scriptstyle#1$}}%
{\mbox{\boldmath$\scriptscriptstyle#1$}}}}

\def\bfmY{\bfmath{Y}}

\def\bfmhhaY{\bfmath{\hhaY}} 
\def\bfmhhaY{\hbox to 0pt{$\widehat{\bfmY}$\hss}\widehat{\phantom{\raise 1.25pt\hbox{$\bfmY$}}}}

\def\rmd{{\rm d}}

\def\til={{\widetilde =}}

\def\clC{{\cal C}}

\def\clN{{\cal N}}

 \def\FRAC#1#2#3{\genfrac{}{}{}{#1}{#2}{#3}}

\def\ddtp{{\mathchoice{\FRAC{1}{d^{\hbox to 2pt{\rm\tiny +\hss}}}{dt}}%
{\FRAC{1}{d^{\hbox to 2pt{\rm\tiny +\hss}}}{dt}}%
{\FRAC{3}{d^{\hbox to 2pt{\rm\tiny +\hss}}}{dt}}%
{\FRAC{3}{d^{\hbox to 2pt{\rm\tiny +\hss}}}{dt}}}}

\def\average#1,#2,{{1\over #2} \sum_{#1}^{#2}}

\def\eye(#1){{\bf(#1)}\quad}

\newtheorem{theorem}{{\bf Theorem}}

\newtheorem{proposition}{{\bf Proposition}}
\newtheorem{corollary}[proposition]{{\bf Corollary}}

\newtheorem{lemma}{{\bf Lemma}}

\def\eq#1/{(\ref{e:#1})}

\newcommand{\beqn}[1]{\notes{#1}%
\begin{eqnarray} \elabel{#1}}

\newcommand{\eeqn}{\end{eqnarray} }

\newcommand{\beq}[1]{\notes{#1}%
\begin{equation}\elabel{#1}}

\newcommand{\eeq}{\end{equation}}

\def\bdes{\begin{description}}
\def\edes{\end{description}}

\newcounter{rmnum}

\newcounter{anum}

{\end{list}}

\def\ass(#1:#2){(#1\ref{#1:#2})}

\def\ritem#1{
\item[{\sf \ass(\current_model:#1)}]
}

\newenvironment{recall-ass}[1]{%
\begin{description}
\def\current_model{#1}}{
\end{description}
}



%% file: macros.tex
\long\def\comment#1{}

\newfont{\bbb}{msbm10 scaled 700}

\newfont{\bb}{msbm10 scaled 1100}
\newcommand{\CC}{\mbox{\bb C}}

\newcommand{\EE}{\mbox{\bb E}}

\newcommand{\av}{{\bf a}}

\newcommand{\cv}{{\bf c}}
\newcommand{\dv}{{\bf d}}

\newcommand{\fv}{{\bf f}}
\newcommand{\gv}{{\bf g}}
\newcommand{\hv}{{\bf h}}

\newcommand{\uv}{{\bf u}}

\newcommand{\xv}{{\bf x}}
\newcommand{\yv}{{\bf y}}
\newcommand{\zv}{{\bf z}}
\newcommand{\zerov}{{\bf 0}}

\newcommand{\Am}{{\bf A}}
\newcommand{\Bm}{{\bf B}}
\newcommand{\Cm}{{\bf C}}

\newcommand{\Fm}{{\bf F}}
\newcommand{\Gm}{{\bf G}}
\newcommand{\Hm}{{\bf H}}
\newcommand{\Id}{{\bf I}}
\newcommand{\Jm}{{\bf J}}

\newcommand{\Lm}{{\bf L}}

\newcommand{\Pm}{{\bf P}}
\newcommand{\Qm}{{\bf Q}}
\newcommand{\Rm}{{\bf R}}

\newcommand{\Tm}{{\bf T}}
\newcommand{\Um}{{\bf U}}
\newcommand{\Wm}{{\bf W}}
\newcommand{\Vm}{{\bf V}}
\newcommand{\Xm}{{\bf X}}

\newcommand{\Ac}{{\cal A}}
\newcommand{\Bc}{{\cal B}}
\newcommand{\Cc}{{\cal C}}

\newcommand{\Ec}{{\cal E}}

\newcommand{\Gc}{{\cal G}}

\newcommand{\Kc}{{\cal K}}
\newcommand{\Lc}{{\cal L}}
\newcommand{\Mc}{{\cal M}}
\newcommand{\Nc}{{\cal N}}

\newcommand{\Pc}{{\cal P}}

\newcommand{\Sc}{{\cal S}}
\newcommand{\Tc}{{\cal T}}

\newcommand{\Vc}{{\cal V}}

\newcommand{\lambdav}{\hbox{\boldmath$\lambda$}}

\newcommand{\Lambdam}{\boldsymbol{\Lambda}}

\newcommand{\Psim}{\boldsymbol{\Psi}}

\newcommand{\trace}{{\hbox{tr}}}
\renewcommand{\arg}{{\hbox{arg}}}

\newcommand{\SNR}{{\sf SNR}}



%% file: newcommands.tex
\newcommand{\Cdlk}{\Cm_k}
\newcommand{\Cdlcirck}{\mathring{\Cm}_k}

\newcommand{\lambdaul}{\lambda_{{\rm ul}}}
\newcommand{\lambdadl}{\lambda_{{\rm dl}}}

\newcommand{\ful}{f_{\rm ul}}
\newcommand{\fdl}{f_{\rm dl}}

\newcommand{\Cul}{\Cm_{\rm ul}}
\newcommand{\Cdl}{\Cm_{\rm dl}}

\newcommand{\Tdl}{{\sf T_{dl}}}

\newcommand{\Pdl}{{\sf P_{dl}}}

\def\herm{{\sfH}}

\def\cg{{\clC\clN}}

\newcommand{\heff}{{\check{\bfh}}_{\rm eff}}
\newcommand{\heffk}{{\check{\bfh}}_{\rm eff}^{(k)}}

\newcommand{\Pthresh}{{\sf P_0}}

\newcommand{\Nvar}{{\sf N_0}}

\newcommand{\Hcoeffs}{\check{\Hm}}